\documentclass{llncs}
\usepackage[T1]{fontenc}
\usepackage{setspace}
\usepackage{amsmath,mathtools,extarrows}
\allowdisplaybreaks[4] 
\usepackage{amssymb}
\usepackage{amsfonts}
\usepackage{bm}
\usepackage{graphicx}
\usepackage{orcidlink}
\usepackage{url}
\hypersetup{colorlinks=true, urlcolor=blue, citecolor=blue, linkcolor=blue}
\usepackage{booktabs,multirow,makecell}
\usepackage{amsfonts}
\usepackage{xcolor}
\usepackage{tcolorbox}
\usepackage{tikz}
\usetikzlibrary{arrows.meta,decorations.pathreplacing}
\usepackage{longtable}
\usepackage{graphicx}
\usepackage{hyperref}
\usepackage{multirow}
\usepackage{makecell}
\usepackage{float}
\usepackage{pdflscape}
\usepackage{caption}
\usepackage{cleveref}
\usepackage{comment}
\usepackage[ruled,vlined]{algorithm2e}
\usepackage{placeins}
\SetKw{KwContinue}{continue}
\SetKw{KwYield}{yield}

\newcommand{\F}{\mathbb{F}}
\newcommand{\Fqm}{\F_{q^m}}
\newcommand{\E}{\mathbb{E}}
\newcommand{\GL}{\mathrm{GL}}
\newcommand{\Span}{\mathrm{Span}}
\newcommand{\rank}{{\mathrm{Rank}}}
\newcommand{\Mat}{\mathcal{M}}
\newcommand{\Gr}{\mathrm{Gr}}

\newcommand{\col}{\mathrm{col}}
\newcommand{\row}{\mathrm{row}}

\crefname{conjecture}{conjecture}{conjectures}
\Crefname{conjecture}{Conjecture}{Conjectures}

\begin{document}
\title{Breaking ACDGV MinRank Gabidulin \\ encryption~schemes over matrix~codes}
\titlerunning{Breaking ACDGV MinRank Gabidulin encryption schemes over matrix codes}
\author{
    Thai Hung Le 
        \inst{1,2}
        \orcidlink{0009-0006-4291-6297}
}

\authorrunning{Thai Hung Le}

\institute{  
        École normale supérieure, PSL University, CNRS, Inria, France\\
        \and 
        LTCI, Telecom Paris, Institut Polytechnique de Paris, France\\
        \email{hung.le@ens.fr} 
}

\maketitle       
\begin{abstract}

Enhanced Gabidulin Matrix Codes (EGMC), introduced by Aragon, Couvreur, Dyseryn, Gaborit, and Vin{\c{c}}otte at Asiacrypt 2024, were designed to hide the algebraic structure of Gabidulin matrix codes while enabling very compact McEliece- and Niederreiter-type encryption schemes, with ciphertexts as small as $65$ bytes at the claimed 128-bit security level.  Their security relies on the assumption that a masked EGMC code is hard to distinguish from a random matrix code.  We show that this enhanced construction leaves enough structure for an equivalent code of the secret key to be recovered.  Unlike previous cryptanalysis, our attack combines combinatorial and algebraic techniques to recover a Gabidulin-equivalent compressed code. This code can then be extended to a full-length equivalent secret key in polynomial time. As a result, the attack provides both a distinguisher and a key-recovery attack against the EGMC encryption schemes. The attack breaks all 16 proposed EGMC parameter sets by large margins. For example, for the claimed 128-bit parameter set $(2,17,37,4,0)$, it reduces the security level from 186 bits to 35 bits and in our implementation, the equivalent secret key of this parameter set is recovered in less than 10 minutes.

\end{abstract}
\section{Introduction}
At Asiacrypt 2024, Aragon, Couvreur, Dyseryn, Gaborit, and Vin{\c{c}}otte introduced a new masking scheme for the Gabidulin matrix code enabling the application of the McEliece and Niederreiter framework to build encryption scheme \cite{aragon2024minrank}. The main idea is to start from an $\F_{q^m}$-linear Gabidulin code $\mathcal C_{\mathrm{vec}}$, expand it to a matrix code $\mathcal C_{\mathrm{mat}}=\Psi_\gamma(\mathcal C_{\mathrm{vec}})\subseteq \F_q^{m\times n}$ via an $\F_q$-basis $\gamma$ of $\F_{q^m}$, then hide the inherited $\F_{q^m}$-structure by appending random extra columns/rows to the basis and scrambling it with invertible matrices $P,Q$, resulting in an Enhanced Gabidulin Matrix Code (EGMC). An adaptation of the McEliece and Niederreiter encryption framework to general matrix codes is introduced by hiding a basis of a secret matrix code, for which there exists an efficient decoding algorithm. The authors finally applied these two frameworks to EGMC and introduced two encryption schemes: EGMC-McEliece and EGMC-Niederreiter. By analyzing possible attacks, the authors finally proposed two set of parameters for EGMC-Niederreiter encryption schemes that offer small ciphertext sizes. Later, at Eurocrypt 2026, Porwal, Wachter-Zeh, and Loidreau introduced an improved key-recovery attack on the schemes \cite{PorwalWL2026ACDGV}. Their attack is based on the kernel attack against the underlying MinRank problem and breaks 5 of the 16 proposed parameter sets. Although the attack remains exponential, it reduces the claimed security of one parameter set targeting 192-bit security to about 161 bits, while that of a 256-bit parameter set is reduced to about 223 bits.

In this work, we introduce a new distinguisher and its extension into a key-recovery attack against the EGMC-McEliece and EGMC-Niederreiter encryption schemes. Unlike the combinatorial distinguisher from \cite{aragon2024minrank}, which guesses both compression (or unscrambling) matrices $U$ and $V$ by brute force, our distinguisher guesses only one of the matrices and uses algebraic techniques to solve for the other one.  We consider two variants of the attack: guess-$V$-solve-$U$, which guesses a column compression $V$ and solves for the row compression $U$, and guess-$U$-solve-$V$, which guesses a row compression $U$ and solves for the full dual space of the hidden compressed Gabidulin code. 

Whenever the distinguisher succeeds, it reveals a punctured code of the secret $\F_{q^m}$-linear code. The second part of the attack extends the punctured code back to the original code in polynomial time by exploiting the structure of Gabidulin codes. Overall, the attack breaks all 16 parameter sets (see \Cref{tab:attack_cost_summary}). For example, an implementation of the attack recovers the equivalent secret keys of the parameter set (2, 17, 37, 4, 0) targeting 128-bit security level within less than $10$ minutes. Details of the implementation of the attack can be found at
\begin{center}
\url{https://gitlab.di.ens.fr/hle/attack_egmc}.
\end{center}

Our attack is different from known attacks as the known distinguishers and key-recovery attacks against EGMC parameters are exponential due to the cost of either brute force or MinRank solver. In such a situation, one may hope to restore security by increasing the parameters.  In contrast, our attack becomes polynomial-time whenever one of the two appending factors $\ell_1=0$ or $\ell_2=0$.  For these parameter families, increasing parameters while keeping either $\ell_1=0$ or $\ell_2=0$ without changing the masking construction does not restore the post-quantum security because the algebraic solving step still reveals the remaining hidden compression matrix in polynomial time. A repair would have to keep both $\ell_1$ and $\ell_2$ nonzero and choose them large enough so that the minimum of the two directions of our attack exceeds the target security level. However, since increasing $\ell_1$ and $\ell_2$ directly increases the dimensions of the public matrix code, and hence the key and ciphertext sizes, whether such a repair remains competitive is still needed to be investigated. Table \ref{tab:attack_cost_summary} summarizes the cost of our attack for both the reference parameters and the alternative parameters of the EGMC-Niederreiter encryption scheme and compares them with the best previous attack as in \cite{aragon2024minrank} and \cite{PorwalWL2026ACDGV}.  
\begin{table*}[htbp]
\centering
\caption{The cost of our attack versus the best previous attacks for the reference and alternate parameter sets. Complexities are given as $\log_2$ of the estimated gate count.  The columns $V\to U$ and $U\to V$ give respectively the guess-$V$-solve-$U$ and the guess-$U$-solve-$V$ costs, with the better of the two shown in bold. The complexities of the best previous attacks are taken from \cite{aragon2024minrank,PorwalWL2026ACDGV}.}
\label{tab:attack_cost_summary}
\scriptsize
\renewcommand{\arraystretch}{1.3}
\setlength{\tabcolsep}{5.7pt}

\begin{tabular}{|c|c|c|c|c|c|c|c|c|c|c|}
\hline
\textbf{Sec.} & \boldmath{$q$} & \boldmath{$k$} & \boldmath{$m$} &
\boldmath{$\ell_1$} & \boldmath{$\ell_2$} & \textbf{pk} & \textbf{ct} &
\boldmath{$V\to U$} &
\boldmath{$U\to V$} &
\makecell{\textbf{Previous best}\\\textbf{attack \cite{aragon2024minrank,PorwalWL2026ACDGV}}} \\
\hline\hline

\multirow{7}{*}{128}
  & 2  & 17 & 37 & 3 & 3 & 76\,kB  & 121\,B & \textbf{88}  & 164  & 165 \\ \cline{2-11}
  & 2  & 25 & 37 & 3 & 3 & 78\,kB  & 84\,B  & \textbf{114} & 164  & 150 \\ \cline{2-11}
  & 2  & 35 & 43 & 2 & 2 & 98\,kB  & 65\,B  & \textbf{108} & 135  & 145 \\ \cline{2-11}
  & 2  & 47 & 53 & 2 & 2 & 166\,kB & 66\,B  & \textbf{134} & 155  & 147 \\ \cline{2-11}
  & 2  & 17 & 37 & 4 & 0 & 70\,kB  & 111\,B & \textbf{35} & 186  & 112 \\ \cline{2-11}
  & 16 & 13 & 23 & 1 & 1 & 41\,kB  & 138\,B & \textbf{82}  & 130  & 148 \\ \cline{2-11}
  & 16 & 7  & 23 & 0 & 5 & 33\,kB  & 207\,B & 183 & \textbf{49}  & 160 \\
\hline \hline

\multirow{4}{*}{192}
  & 2 & 51 & 59 & 2 & 2 & 268\,kB & 89\,B  & \textbf{142} & 170 & 209 \\ \cline{2-11}
  & 2 & 23 & 43 & 5 & 0 & 133\,kB & 134\,B & \textbf{38} & 255 & 161 \\ \cline{2-11}
  & 2 & 33 & 47 & 5 & 0 & 173\,kB & 111\,B & \textbf{41} & 275 & 221 \\ \cline{2-11}
  & 2 & 41 & 53 & 4 & 0 & 230\,kB & 106\,B & \textbf{41} & 253 & 212 \\
\hline \hline

\multirow{5}{*}{256}
  & 2  & 23 & 47 & 3 & 3 & 191\,kB & 177\,B & \textbf{108} & 198 & 240 \\ \cline{2-11}
  & 2  & 37 & 53 & 3 & 2 & 274\,kB & 139\,B & \textbf{115} & 214 & 262 \\ \cline{2-11}
  & 2  & 71 & 79 & 2 & 2 & 667\,kB & 119\,B & \textbf{185} & 213 & 289 \\ \cline{2-11}
  & 16 & 9  & 29 & 2 & 1 & 87\,kB  & 334\,B & \textbf{69}  & 272 & 223 \\ \cline{2-11}
  & 16 & 17 & 29 & 2 & 1 & 107\,kB & 218\,B & \textbf{103} & 273 & 304 \\
\hline
\end{tabular}
\end{table*}

The rest of the paper is organized as follows.  \Cref{sec:preliminaries} recalls some preliminaries, the EGMC masking transformation, and the known attacks considered in the original article and in \cite{PorwalWL2026ACDGV}. \Cref{sec:hybrid_distinguisher} introduces the hybrid distinguisher in its two directions: guess-$V$-solve-$U$ and guess-$U$-solve-$V$.  \Cref{sec:code_extension} shows how to extend the recovered punctured code to an equivalent full secret Gabidulin code.  \Cref{sec:decoding_ciphertext} explains how this equivalent trapdoor decrypts both EGMC-McEliece and EGMC-Niederreiter ciphertexts.  \Cref{sec:prob_rank_def_U} shows the arguments for the rank-$1$ solving step, and \Cref{sec:complexity} computes the cost for the concrete parameter sets.

\section{Preliminaries}
\label{sec:preliminaries}
\subsection*{Notation}
\begin{itemize}
\item We use ordinary letters for scalars, vectors, matrices, and calligraphic letters for vector spaces.
\item For vectors $x,y$, $\langle x,y\rangle$ denotes their inner product $\sum_i x_i y_i$.
\item For a vector $v$, $v[i]$ denotes its $i$-th coordinate. For a matrix $M$, $M[i,j]$ denotes its entry at $i$-th row and $j$-th column.
\item Subscripts such as $v_i$ and $M_i$ denote a vector and a matrix in a families of vectors and matrices.
\item For matrices $A$ and $B$ with the same number of rows, $[A\mid B]$
denotes their horizontal concatenation. The same convention is used for
$[A_1\mid\cdots\mid A_s]$.  
\item $M^\top$ denotes the transpose of a matrix $M$. 
\item $\mathbf V(I)$ denotes the variety of the ideal $I$. 
\item For any field $\mathbb{K}$, positive integers $a,b$, and $M=(M[i,j])$ we define the vectorization map that sends a matrix to a vector as
\begin{align*}
\operatorname{vec}_{a,b} : \mathbb{K}^{a\times b} &\longrightarrow \mathbb{ K}^{a b}  \\
M &\longmapsto 
(M[1,1],\dots,M[1,b],M[2,1],\dots,M[a,b]).
\end{align*}
The reshaping map is defined as the inverse of the vectorization map \[\operatorname{mat}_{a,b} = \operatorname{vec}_{a,b} ^{-1}.\]
\item We define $\sigma:\F_{q^m} \to \F_{q^m}$ to be the Frobenius automorphism $\sigma(x)=x^q$.
\item For a vector $v$, $v^{[q^i]}$ denotes the vector obtained by applying $\sigma$ entrywise to $v$. For a matrix $V$, $V^{[q^i]}$ denotes the matrix obtained by applying $\sigma$ entrywise to $V$.
\item Two matrix codes
$\mathcal A,\mathcal B\subseteq \F_q^{a\times b}$ are called
\emph{equivalent} if there exist $L\in\GL_a(\F_q)$ and
$R\in\GL_b(\F_q)$ such that $\mathcal B=L\mathcal A R$.
\item For a positive integer $n$, we denote by
$
\mathrm{GL}_n(\mathbb{F}_q)
=
\{ A \in \mathbb{F}_q^{n \times n} : \det(A) \neq 0 \}
$
the general linear group of invertible $n \times n$ matrices over $\mathbb{F}_q$. Its size is
$
\left|\mathrm{GL}_n(\mathbb{F}_q)\right|
=
\prod_{i=0}^{n-1} (q^n-q^i).
$

\item For integers $0 \le r \le n$, we denote by
$
\mathrm{Gr}(r,n)
=
\{ \mathcal U \subset \mathbb{F}_q^n : \dim_{\mathbb{F}_q}(\mathcal U)=r \}
$
the Grassmannian of $r$-dimensional linear subspaces of $\mathbb{F}_q^n$. Its size is the Gaussian binomial coefficient
$
\left|\mathrm{Gr}(r,n)\right|
=
\binom{n}{r}_q
=
\prod_{i=0}^{r-1}
\frac{q^n-q^i}{q^r-q^i}.
$
\end{itemize}
\subsection{Gabidulin Codes}
Gabidulin codes \cite{delsarte1978bilinear,gabidulin1973} are a class of rank-metric codes that can be defined as an evaluation code on a class of $q$-polynomial of $q$-degree $r$ in $\F_{q^m}[x]$ of the form
\[
P(x)= a_0 x + a_1 x^q + a_2 x^{q^2} + \dots + a_{r} x^{q^{r}}.
\]
We denote by $\F_{q^m}[x]_{< k}$ the set of $q$-polynomials of $q$-degree strictly less than $k$. Let $k \le n \le m$ be integers, and let $ g = \{g_1, \dots, g_n\}$ be a set of $n$ linearly independent elements of $\F_{q^m}$ with respect to $\F_q$, then the Gabidulin code of length $n$, dimension $k$, evaluation vector $g$, denoted $\mathcal{G}_g(n,k,m)$ is defined by:
\[
\mathcal{G}_g(n,k,m) = \left\{  \left(P(g_1), ..., P(g_n)\right) \mid P \in \F_{q^m}[x]_{< k} \right\}.
\]
 
\subsection{Gabidulin Matrix Codes}
Let $\gamma = (\gamma_1, \dots, \gamma_m)$ be an $\mathbb{F}_q$-basis of $\mathbb{F}_{q^m}$. The $\gamma$-expansion of an element in $\mathbb{F}_{q^m}$ to a vector in $\mathbb{F}_q^m$ is defined as the map:
\begin{align*}
\Psi_{\gamma} :  \mathbb{F}_{q^m} &\longrightarrow  \mathbb{F}_q^m\\
x = \sum_{i=1}^m x_i \gamma_i  &\longmapsto (x_1, \dots, x_m)
\end{align*}
\noindent
We can extend $\Psi_{\gamma}$ naturally coordinate-wise to a word
$x=(x_1,\dots,x_n)\in \mathbb{F}_{q^m}^n$ by reading the expanded coordinates
as columns:
\[
\Psi_{\gamma}(x)
=
\bigl[\Psi_{\gamma}(x_1)\mid\cdots\mid\Psi_{\gamma}(x_n)\bigr]
\in \mathbb{F}_q^{m \times n}.
\]
Its inverse therefore identifies a matrix in $\mathbb{F}_q^{m\times n}$ with
the corresponding word in $\mathbb{F}_{q^m}^n$, column by column. When the
basis $\gamma$ is fixed, we simply write $\Psi$ for $\Psi_\gamma$. Then a
Gabidulin code can be turned into a Gabidulin matrix code of size $m \times n$
and dimension $mk$ by applying the $\gamma$-extension. However, different
$\gamma$-extensions of a Gabidulin code produce left-equivalent Gabidulin
matrix codes (see \cite{couvreur2020hardness}).

It is well-known that Gabidulin codes reach the Singleton bound, i.e. they have
minimum rank distance $n-k+1$; and have efficient decoding algorithm up to the unique decoding radius $\frac{n-k}{2}$ \cite{loidreau2005welch}. However, Gabidulin codes can easily be distinguished from random $\Fqm$-linear codes, and therefore cannot be directly used inside a McEliece-like cryptosystem \cite{couvreur2023extension,overbeck2008structural}.
\subsubsection{Overbeck distinguisher for Gabidulin matrix codes}
The Gabidulin code, as a $\F_{q^m}$-linear code, can be distinguished from a $\F_{q}$-linear random code in polynomial time by computing the left stabilizer algebra \cite{couvreur2020hardness}. In this section, we explain why a Gabidulin code can also be distinguished from a random $\F_{q^m}$-linear code in polynomial time using the Overbeck distinguisher \cite{overbeck2008structural}. 

Let $\mathcal C \subseteq \F_{q^m}^n$ be an $\F_{q^m}$-linear $[n,k]$ code, we define the
\emph{Frobenius sum}
\[
\Lambda_f(\mathcal C) \;=\; \mathcal C + \mathcal C^{[q^1]} + \cdots + \mathcal C^{[q^f]} .
\]
Let $G$ is a generator matrix of $\mathcal C$ then a
generator matrix of $\Lambda_f(\mathcal C)$ is obtained by stacking these
Frobenius powers by rows:
\[
G_{\Lambda,f}
=
\begin{bmatrix}
G\\
G^{[q^1]}\\
\vdots\\
G^{[q^f]}
\end{bmatrix}.
\]
For a Gabidulin code $\mathcal C_{\mathrm{Gab}}$, one has
$
\dim_{\F_{q^m}}\Lambda_f(\mathcal  C_{\mathrm{Gab}})=\min\{n,k+f\},
$
whereas for a random $[n,k]$ code one typically has
$
\dim_{\F_{q^m}}\Lambda_f(\mathcal  C)=\min\{n,k(f+1)\},
$
because the $f+1$ Frobenius shifts are generically $\F_{q^m}$-linearly independent.
This gap yields a polynomial-time distinguisher and forms the basis of Overbeck's
structural key-recovery: by choosing $f$ appropriately, one reconstructs a larger
Gabidulin code (with the same evaluation vector) from a suitable subcode and then
recovers the evaluation vector $g$, obtaining an equivalent secret key. Computing $\dim \Lambda_f(\mathcal C)$ amounts to Gaussian elimination on an
$((f+1)k)\times n$ matrix over $\F_{q^m}$, with cost
$
\tilde {\mathcal O}\!\big(((f+1)k)^2\,n\big)
$ operations in $\F_{q^m}$,
and in the common regime $(f+1)k=\Theta(n)$ this is $\tilde {\mathcal O}(n^3)$. The subsequent
recovery of the Gabidulin support/evaluation vector can be done in polynomial time,
often stated as
$
\mathcal O(n^3)
$ operations in $\F_{q^m}$.

\subsection{Enhanced Gabidulin Matrix Codes (EGMC)}
In \cite{aragon2024minrank}, a new transformation, named ``Random rows and columns matrix code transformation'', was introduced to mask a secret Gabidulin matrix code. Its main goal is to prevent the Overbeck attack mentioned previously.
\begin{definition}[Random Rows and Columns matrix code transformation \cite{aragon2024minrank}]
Let $m,n,K,\ell_1,\ell_2 \in \mathbb{N}$. Let $\mathcal{C}_{mat}$ be a matrix code of size $m \times n$ and dimension $K$ over $\F_q$. Let $\mathcal{B} = (\mathbf{A}_1, \dots, \mathbf{A}_K)$ be a basis of $\mathcal{C}_{mat}$. The Random Rows and Columns matrix code transformation consists in sampling uniformly at random the following matrices: $\mathbf{P} \xleftarrow{\$} \mathbf{GL}_{m+\ell_1}(\F_q)$, $\mathbf{Q} \xleftarrow{\$} \mathbf{GL}_{n+\ell_2}(\F_q)$ and $3K$ random matrices: $\mathbf{R}_i \xleftarrow{\$} \F_q^{m \times \ell_2}$, $\mathbf{R}'_i \xleftarrow{\$} \F_q^{\ell_1 \times n}$ and $\mathbf{R}''_i \xleftarrow{\$} \F_q^{\ell_1 \times \ell_2}$; and define the matrix code whose basis is:
\[
\mathcal{RB} = \left( \mathbf{P} \begin{bmatrix} \mathbf{A}_1 & \mathbf{R}_1 \\ \mathbf{R}'_1 & \mathbf{R}''_1 \end{bmatrix} \mathbf{Q}, \dots, \mathbf{P} \begin{bmatrix} \mathbf{A}_K & \mathbf{R}_K \\ \mathbf{R}'_K & \mathbf{R}''_K \end{bmatrix} \mathbf{Q} \right).
\]
\end{definition}

An Enhanced Gabidulin matrix code $\mathcal{EG}_{\boldsymbol{g}}(n,k,m,\ell_1,\ell_2)$ is the matrix code $\Psi_\gamma(\mathcal{G})$ on which we apply the Random Rows and Columns matrix code transformation above.

\begin{definition}[EGMC($k,m,n,\ell_1,\ell_2$) distribution \cite{aragon2024minrank}]
\label{def:egmc_distribution}
Let $k, m, n, \ell_1, \ell_2 \in \mathbb{N}$ be such that $k \le n \le m$. The Enhanced Gabidulin Matrix Code distribution EGMC($k,m,n,\ell_1,\ell_2$) sample $\mathbf g=(g_1,\dots,g_n)$ uniformly at random among all $n$-tuples of
$\F_{q^m}$ that are $\F_q$-linearly independent,
, a basis $\gamma \xleftarrow{\$} \mathcal{B}(\mathbb{F}_{q^m})$, $\mathbf{P} \xleftarrow{\$} \mathbf{GL}_{m+\ell_1}(\mathbb{F}_q)$, $\mathbf{Q} \xleftarrow{\$} \mathbf{GL}_{n+\ell_2}(\mathbb{F}_q)$ and $3km$ random matrices: $\mathbf{R}_i \xleftarrow{\$} \mathbb{F}_q^{m \times \ell_2}$, $\mathbf{R}'_i \xleftarrow{\$} \mathbb{F}_q^{\ell_1 \times n}$ and $\mathbf{R}''_i \xleftarrow{\$} \mathbb{F}_q^{\ell_1 \times \ell_2}$; computes $\mathcal{B} = (\mathbf{A}_1, \dots, \mathbf{A}_{km})$ a basis of the matrix code $\Psi_\gamma(\mathcal{G})$, and outputs the matrix code with basis:
\[
\mathcal{RB} = \left( \mathbf{P} \begin{bmatrix} \mathbf{A}_1 & \mathbf{R}_1 \\ \mathbf{R}'_1 & \mathbf{R}''_1 \end{bmatrix} \mathbf{Q}, \dots, \mathbf{P} \begin{bmatrix} \mathbf{A}_{km} & \mathbf{R}_{km} \\ \mathbf{R}'_{km} & \mathbf{R}''_{km} \end{bmatrix} \mathbf{Q} \right).
\]
\end{definition}  

\begin{definition}[EGMC-Indistinguishability problem \cite{aragon2024minrank}]
Given a matrix code $\mathcal{C}_{mat}$ of size $(m+\ell_1) \times (n+\ell_2)$ and dimension $mk$, the Decisional EGMC-Indistinguishability $(k, m, n, \ell_1, \ell_2)$ problem asks to decide with non-negligible advantage whether $\mathcal{C}_{mat}$ is sampled from the EGMC$(k, m, n, \ell_1, \ell_2)$ distribution or the uniform distribution over the set of $\mathbb{F}_q$-subspaces of $\mathbb{F}_q^{(m+\ell_1) \times (n+\ell_2)}$ of dimension $mk$.
\end{definition}

\begin{definition}[EGMC-Search problem \cite{aragon2024minrank}] Let $k, m, n, \ell_1, \ell_2 \in \mathbb{N}$ be such that $k \le n \le m$, and $\mathcal{C}_{mat}$ sampled from the EGMC($k,m,n,\ell_1,\ell_2$) distribution. The EGMC-Search problem asks to retrieve the basis $\gamma \in \mathcal{B}(\mathbb{F}_{q^m})$ and the evaluation vector $\mathbf{g} \in \mathbb{F}_{q^m}^n$ used to construct $\mathcal{C}_{mat}$.
\end{definition}

It should be noted that the outputs of this search problem are not unique as changing $\gamma$ corresponds to left multiplication by an invertible matrix over $\mathbb{F}_q$; changing the evaluation vector by an $\mathbb{F}_q$-linear transformation can produce equivalent Gabidulin descriptions; and, the masking by $(P,Q)$ further destroys uniqueness. We can say that what is normally and necessarily recoverable is an \emph{equivalent trapdoor}, not the original $(\gamma,\mathbf g)$.

\subsection{EGMC-McEliece and EGMC-Niederreiter encryption schemes}
In \cite{aragon2024minrank}, the authors proposed two encryption schemes whose the public key is an Enhanced Gabidulin Matrix Code, and the secret key is the original Gabidulin matrix code and two scambling matrices (see Appendix \ref{appendix:McEliece_Niederreiter_encryption_framework}). The security of these encryption schemes is based on the conjectured difficulty of the EGMC-Indistinguishability problem. 

\subsection{Previous attacks on EGMC}
This subsection recalls the attacks considered in the security analysis of EGMC \cite{aragon2024minrank} and in the cryptanalysis by Porwal, Wachter-Zeh, and Loidreau \cite{PorwalWL2026ACDGV}. All of these attacks have exponential complexity.

\subsubsection{The combinatorial distinguisher of EGMC}
In \cite{aragon2024minrank}, the authors introduced a combinatorial distinguisher that solves the EGMC-Indistinguishability problem with the complexity of 
\begin{equation*}
\widetilde{\mathcal{O}}\left(q^{m \ell_1 + (k+1) \ell_2}\right).
\end{equation*}
The idea is to apply a projection map on both the row and column spaces of the public code $\mathcal{C}$ to get rid of the contributions of the matrices $R_i$, $R'_i$ and $R''_i$ while not destroying the underlying
$\F_{q^m}$-linear structure. Likewise, in the terminology used throughout the rest of the paper, these projection matrices are the row and column compression matrices: \emph{the row compression is the matrix multiplying public codewords on the left, and the column compression is the matrix multiplying them on the right}. We guess two such compression matrices: a full-rank row compression $U = [\,U_0 \mid 0\,] \in \F_q^{m\times (m+\ell_1)}$, where $U_0 \in \GL_m(\F_q)$; and for some integer $n'$ such that $k < n' \le n$, a full-rank column compression
$V =
\begin{bmatrix}
V_0\\
0
\end{bmatrix} \in \F_q^{(n+\ell_2)\times n'}$ where $V_0 \in \F_q^{n\times n'}$. 

The resulting compressed code is $U\mathcal C V=\{UMV:M\in\mathcal C\}$. Then we can distinguish any $\F_{q^m}$-linear code, and hence $\mathcal{C}$, from a random code by computing the dimension of the left stabilizer of $U\mathcal C V$. We need $n'>k$ because if $n'\le k$, then any $k$-dimensional
$\F_{q^m}$-linear code of length $n'$ coincides with the whole ambient
space $\F_{q^m}^{n'}$. Under matrix representation this is
$\F_q^{m\times n'}$, whose left stabilizer algebra is the full matrix
algebra $\mathcal{M}_m(\F_q)$. In that regime the method can not distinguish. Therefore the smallest informative choice is $n'=k+1$.
\subsubsection{Overbeck-like distinguishers of EGMC}
In EGMC-type constructions, the Frobenius map in the original Overbeck attack \cite{overbeck2008structural} is
\emph{not directly available} on the public matrix code, so an Overbeck attack cannot
be applied directly. However, an Overbeck-like distinguisher has been proposed in the security analysis of EGMC: One can use the fact that for a hidden Gabidulin matrix code there exist nontrivial
      multipliers $D$ such that the span of $C + DC$ fails to fill the ambient space,
      while for a random code it typically does. This leads to a bilinear system in
      unknowns $(D,M)$ with constraints of the form $\mathrm{Tr}(DCM^\top)=0$ for all
      public basis matrices $C$; naive linearization yields too many variables
      (on the order of $m^4$) for only $\Theta(m^2)$ equations, suggesting infeasibility
      for the targeted parameters.

\subsubsection{Key Recovery Attack based on the kernel attack of the Minrank problem \cite{PorwalWL2026ACDGV}}
At Eurocrypt 2026, Porwal, Wachter-Zeh and Loidreau introduced a Minrank-based key-recovery attack against EGMC. The cost of the attack is still exponential but it breaks 5 over 16 parameter sets of EGMC. The attack exploits the fact that the hidden Gabidulin matrix code has a large left-stabilizer algebra, whereas a random public-looking matrix code has stabilizer dimension only
$1$.  For public generators
\[
B_i=P\begin{bmatrix}A_i&*\\ *&*\end{bmatrix}Q,
\]
the attacker searches for a matrix $F$ and public codewords
$C_i\in \mathcal C_{\mathrm{pub}}$ such that
\[
F B_i-C_i
=
P\begin{bmatrix}0&*\\ *&*\end{bmatrix}Q .
\]
This is done by guessing left and right kernel vectors $v,w$ and using
the linear equations
\[
v^\top\left(FB_i-\sum_j \mu_{i,j}B_j\right)w=0
\]
to solve for $F$ and the coefficients $\mu_{i,j}$.  Once such
relations are found, they reveal enough information about the hidden row
and column transformations $P,Q$ to remove the appended random rows and
columns, thereby recovering a code equivalent to
$\mathcal C_{\mathrm{mat}}$. 
\subsubsection{Message attacks}
Two classes of attacks on the message are analyzed in the security analysis of EGMC in \cite{aragon2024minrank}: generic rank decoding attacks and MinRank attacks. These attacks do not exploit the hidden Gabidulin structure in the same way as the structural distinguishers above; they are recalled here to explain the baseline security model used for the proposed parameters.
\paragraph{Generic rank decoding.}
Given a ciphertext $y = xG_{\mathrm{pub}} + z$ with $\rank(z)=r$, an attacker
without the key faces the rank-metric decoding problem for a length-$n$ code over
$\mathbb{F}_{q^m}$. Generic rank decoding algorithms have exponential cost in the
security parameters; historically, this was one motivation for rank-metric McEliece
variants (before structural breaks).
\paragraph{Generic MinRank.}
For matrix-code McEliece/Niederreiter frames, decrypting without the secret key
reduces to solving the following MinRank instance from public matrices $M_i$ spanning the public code :
$
\rank (Y - \sum_{i=1}^{K} x_i M_i) \le r.
$
These attacks have exponential costs for proposed parameters of EGMC.

\section{Hybrid Distinguisher for EGMC}
\label{sec:hybrid_distinguisher}
This section introduces a new ``hybrid'' distinguisher for the EGMC Indistinguishability Problem. Let $k\le n\le m$ and
$\mathcal C\subset \F_q^{(m+\ell_1)\times(n+\ell_2)}$ be the public $\F_q$-subspace of dimension $mk$, which is EGMC in our setting. Let $\mathcal U$ be the set of full-rank matrices
$U\in\F_q^{m\times(m+\ell_1)}$, and let $\mathcal V$ be the set of
full-rank matrices $V\in\F_q^{(n+\ell_2)\times(k+1)}$. For a chosen pair of matrices $(U,V)\in\mathcal U\times\mathcal V$, define
\begin{align*}
\delta_{U,V}:\F_q^{(m+\ell_1)\times(n+\ell_2)}&\longrightarrow \F_q^{m\times(k+1)}
\\
X&\longmapsto UXV,
\end{align*}
and we set
$
\mathcal S_{U,V} = \delta_{U,V}(\mathcal C)\subseteq \F_q^{m\times(k+1)}
$ the compressed code of $\mathcal C$ by $U$ and $V$.

The hybrid distinguisher is based on the following main idea: \emph{instead of guessing both a row compression and a column compression, one guesses one compression and solves algebraically for the other one}.  

We present two directions.  The first is the guess-$V$-solve-$U$ direction, where one guesses a column compression $V\in \F_q^{(n+\ell_2)\times(k+1)}$ and solves for a row compression $U\in \F_q^{m\times(m+\ell_1)}$ such that the compressed code $\mathcal S_{U,V}$ is \emph{equivalent} to a matrix representation of a $[k+1,k]$ Gabidulin code over $\F_{q^m}$.  The second is the guess-$U$-solve-$V$ direction, where one guesses $U$ and solves for a full column compression $V$ with a similar goal in mind.  

\begin{remark}
  To be more precise about the code equivalence property, from now on we work with an equivalence class of $(U,V)$ where $(U',V') \sim (U,V)$ when there exist $ G\in\GL_m(\F_q),\; H\in\GL_{k+1}(\F_q)$ such that $(U',V') = (GU,VH)$. $\mathcal S_{U,V}$ is
equivalent to a matrix representation of a $[k+1,k]$ Gabidulin code means that it is a $[k+1,k]$ Gabidulin code up to left and right multiplication by invertible matrices over $\F_q$. The specific subspace $\mathcal S_{U,V}$ depends on the
choice of representative $U,V$, but this equivalence property depends only on the equivalence class of the pair.
\end{remark}
\subsection{Guess-$V$-solve-$U$ direction}
\label{subsec:guessVsolveU}
\subsubsection{Phase 1: Guessing $V$.}
We treat $V$ as the ``guess'' variable. We call $V$ \emph{valid} if there exists a matrix $U$ such that the pair $(U,V)$ compresses $\mathcal{C}$ into a subspace $\mathcal{S}_{U,V}$ such that $\mathcal{S}_{U,V}$ is equivalent to a $[k+1,k]$ Gabidulin code. 

Due to the construction of the EGMC, we know that if
$ V = Q^{-1} V'$ and $V'=[(V'')^\top\mid 0]^\top $ where $V'' \in \F_q^{n\times (k+1)}$, $Q\in\GL_{n+\ell_2}(\F_q)$ (as in Definition~\ref{def:egmc_distribution}), then there exists a pair $(U,V)$ compressing $\mathcal{C}$ into the desired Gabidulin code. Therefore, the probability of guessing a valid $V$ among random full-rank matrices of size $(n+\ell_2)\times(k+1)$ is
\[
\frac{|\Gr(k+1,n)|}{|\Gr(k+1,n+\ell_2)|}
=\frac{\binom{n}{k+1}_q}{\binom{n+\ell_2}{k+1}_q}
\approx q^{-(k+1)\ell_2}.
\]

\subsubsection{Phase 2: Algebraic Modeling.}
We first describe the overall idea before diving into details.  For a valid guess
$V$, the compressed code $\mathcal S_{U,V}$ should become, after the
column-wise identification with $\F_{q^m}^{k+1}$, an
$\F_{q^m}$-hyperplane.  Thus it is annihilated by some nonzero vector
$h\in\F_{q^m}^{k+1}$. We model this condition with the unknown row
compression $U$ and the unknown hyperplane vector $h$, obtaining a
bilinear system in the entries of $h$ and in the columns of $U$.  We then
linearize the system and compute its kernel.  Since this linear system is
underdetermined, its kernel contains many linearized solutions. Yet, the solutions
coming from an actual pair $(h,U)$ are precisely the kernel elements whose
product matrix $W=(w_{j,s})$ has rank $1$. We now explain the algebraic model in more details.

Using the coordinate-wise extension of $\Psi=\Psi_\gamma$, let
$\col_{\F_{q^m}}$ denote the inverse word expansion on $k+1$ columns:
\begin{align*}
\col_{\F_{q^m}} = \left(\Psi^{-1}\right)^{\oplus k+1}:\F_q^{m\times(k+1)} &\longrightarrow \F_{q^m}^{k+1} \\
\big([x_1|\cdots|x_{k+1}]\big)
&\longmapsto
\big(\Psi^{-1}(x_1),\dots,\Psi^{-1}(x_{k+1})\big).
\end{align*} For a valid $V$, there exists $U$ such that $\mathcal{S}_{U,V}$ has dimension $mk$ since it is equivalent to a $[k+1,k]$ Gabidulin matrix code. Here, equivalence means there exists another basis $\gamma'$ of $\F_{q^m}$ over $\F_q$ such that $\mathcal{S}_{U,V}=\Psi_{\gamma'}(\mathcal{D})$ for an $\F_{q^m}$-linear code $\mathcal{D}\subseteq \F_{q^m}^{k+1}$ of dimension $k$. However, as the code is equivalence up to the left-multiplication by a matrix in $\GL_m(\F_q)$, we may assume $\gamma'=\gamma$. Therefore, we have the following proposition: 
\begin{proposition}
\label{prop:main_prop}
If $\mathcal S_{U,V}$ is equivalent to the matrix representation of a
$[k+1,k]$ Gabidulin code over $\F_{q^m}$ then there exists a nonzero vector
    $h\in\F_{q^m}^{k+1}$, unique up to $\F_{q^m}^\times$-scaling, such
    that
    \[
    \col_{\F_{q^m}}(\mathcal{S}_{U,V})
    =
    \{x\in\F_{q^m}^{k+1}:\langle h,x\rangle=0\}.
    \]
\end{proposition}

\begin{proof}
We have $\mathcal{S}_{U,V}=\Psi_\gamma(\mathcal{D})$ for an $\F_{q^m}$-linear code $\mathcal{D}\subseteq \F_{q^m}^{k+1}$ of dimension~$k$. That implies $\col_{\F_{q^m}}(\mathcal{S}_{U,V})=\mathcal{D}$.
  In other words, $\col_{\F_{q^m}}(\mathcal{S}_{U,V})$ is an $\F_{q^m}$-hyperplane in $\F_{q^m}^{k+1}$, so there exists a nonzero
  $h\in \F_{q^m}^{k+1}$, unique up to
  $\F_{q^m}^\times$-scaling, such that
  $
  \col_{\F_{q^m}}(\mathcal{S}_{U,V})
  =
  \{x\in \F_{q^m}^{k+1}:\langle h,x\rangle=0\}.
  $ 
\end{proof}
Using \Cref{prop:main_prop}, the problem of finding $U$ is equal to finding $h \in \F_{q^m}^{k+1} \setminus \{0\}$ such that 
\begin{align}
  \label{eq:main}
  \left\langle h,\ \col_{\F_{q^m}}(\mathcal{S}_{U,V})\right\rangle=0.
\end{align} 

For a fixed basis $\{c_1,\dots,c_{mk}\}$ of $\mathcal{C}$, \Cref{eq:main} has $(k+1) \times (m + \ell_1)$ bilinear variables over $\F_{q^m}$ and $mk$ equations whose coefficients are over $\F_q$. Indeed, we set 
\begin{align*}
s_i &= U c_i V \in \mathcal{S}_{U,V} \quad \forall \; i=1,\dots,mk\\
D_i &=c_iV\in \F_q^{(m+\ell_1)\times(k+1)} \quad \forall \; i=1,\dots,mk.
\end{align*}
\noindent
We write $D_i$ as $[d_{i,1}|\cdots|d_{i,k+1}]$ with $d_{i,j}\in \F_q^{m+\ell_1}$ and $U\in \F_q^{m\times(m+\ell_1)}$ by columns $U=[u^{(\mathrm{vec})}_1|\cdots|u^{(\mathrm{vec})}_{m+\ell_1}]$ with $u^{(\mathrm{vec})}_s\in \F_q^m$. Then, we convert these columns $u^{(\mathrm{vec})}$ to a vector $u\in\F_{q^m}^{m+\ell_1}$ with the coordinates
\[
u[s]=\Psi^{-1}\big(u^{(\mathrm{vec})}_s\big)\in \F_{q^m} \quad \forall \; s \in 1,\dots,m+\ell_1.
\]
Then for each $(i,j)$ :
\begin{align*}
\Psi^{-1}(U d_{i,j})
&=\Psi^{-1}\left(\sum_{s=1}^{m+\ell_1} d_{i,j}[s] u^{(\mathrm{vec})}_s\right) \\
&=\sum_{s=1}^{m+\ell_1} d_{i,j}[s]  \Psi^{-1}\left(u^{(\mathrm{vec})}_s\right) \text{ because $d_{i,j}[s]\in \F_q$ and $\Psi$ is $\F_q$-linear }\\
&= \sum_{s=1}^{m+\ell_1} d_{i,j}[s]  u[s]
\; \in \; \F_{q^m}.
\end{align*}
Then we have
\begin{align}
    \label{eq:hyperplaneconstraint}
    \Cref{eq:main} &\iff 
\sum_{j=1}^{k+1} h[j]\cdot \Psi^{-1}(U d_{i,j}) = 0\notag\\
&\iff
\sum_{j=1}^{k+1}\sum_{s=1}^{m+\ell_1} d_{i,j}[s] (h[j] u[s])=0.
\end{align}
Now, we introduce the bilinear products
\[
w[j,s]=h[j]u[s]\in \F_{q^m} \; \text{where }
\; 1\le j \le k+1,\ 1\le s\le m+\ell_1.
\]
Then the constraints (\ref{eq:hyperplaneconstraint}) become the linear system over $\F_{q^m}$:
\begin{equation}
\sum_{j=1}^{k+1}\sum_{s=1}^{m+\ell_1} d_{i,j}[s] w[j,s]=0
\;\; \text{for } i=1,\dots,mk,
\label{eq:linearized-Fqm}
\end{equation}
where all coefficients $d_{i,j}[s]$ lie in the subfield $\F_q$. Let
$
N_w=(k+1)(m+\ell_1)
$
and set
\[
w= \bigl(w[j,s]\bigr)_{\substack{1\le j\le k+1 \\ 1\le s\le m+\ell_1}}
\in \F_{q^m}^{N_w}.
\]
Then \eqref{eq:linearized-Fqm} is a homogeneous linear system of $mk$ equations, with all coefficients in $\F_q$, while the $N_w$ unknowns are in $\F_{q^m}$.

Let $A(V)\in \F_q^{mk\times N_w}$ be the coefficient matrix of \eqref{eq:linearized-Fqm} as these coefficents depend on the guess of $V$. In more details, its entries are the scalars $d_{i,j}[s]\in \F_q$ placed in the obvious way.
Since all entries of $A(V)$ lie in $\F_q$, its right kernel over $\F_{q^m}$ is obtained from its right kernel over $\F_q$ by scalar extension
$
\ker_{\F_{q^m}} A(V)=\F_{q^m}\otimes_{\F_q}\ker_{\F_q}A(V)
$,
identified with the $\F_{q^m}$-solution space by
$\alpha\otimes x\mapsto \alpha x$,
so the same $\F_q$-basis is also a basis of $\ker_{\F_{q^m}}A(V)$ after scalar
extension. However, this does not mean that the kernel elements we seek have entries
in $\F_q$ since the unknowns $w[j,s]$ are still
$\F_{q^m}$-valued. It only means that, because the linear equations are
defined over the subfield $\F_q$, the $\F_{q^m}$-solution space admits a
basis with entries in $\F_q$.  

In more details, if $A(V)$ has full row rank $mk$ over $\F_q$ (equivalently
over $\F_{q^m}$), then
\[
\rho =\dim_{\F_{q^m}}\ker A(V) = N_w - mk = m+(k+1)\ell_1 \;.
\]
For an $\F_q$-basis $b_1,\dots,b_\rho\in\F_q^{N_w}$ of
$\ker_{\F_q}A(V)$, we set
\[
K_t=\operatorname{mat}_{k+1,m+\ell_1}(b_t)
\in\F_q^{(k+1)\times(m+\ell_1)}
\quad\text{for }t=1,\dots,\rho.
\]

\noindent
Because the variables originally come from the bilinear products $w[j,s]=h[j]u[s]$, any element of $\ker_{\F_{q^m}}A(V)$ can then be
written uniquely with extension-field coefficients as
\[
W(\alpha)=\sum_{t=1}^{\rho}\alpha[t] K_t
\;\; \text{where } \alpha=(\alpha[1],\dots,\alpha[\rho])\in \F_{q^m}^{\rho}.
\]
Recovering a valid pair $(h,u)$ therefore reduces to finding a
nonzero kernel vector of $A(V)$ whose reshaped matrix (via the reshaping map) has rank $1$. Thus, we seek a nonzero $W(\alpha)$ such that
\begin{equation}
    \label{minrankinstance}
\rank_{\F_{q^m}}(W(\alpha))=1.
\end{equation}
\noindent
The rank-$1$ condition~\eqref{minrankinstance} is equivalent to the vanishing of all $2\times 2$ minors of $W(\alpha)$. The quadratic system generated from these $2\times 2$ minors has
\[
  N_{\mathrm{eq}}
    = \tbinom{k+1}{2}\,\tbinom{m+\ell_1}{2}
\]
\emph{homogeneous quadratic} equations in the $\rho$ coordinates
$\alpha[1],\dots,\alpha[\rho]$ over $\F_{q^m}$. 
\begin{remark} Any rank-$1$ kernel vector $W=h u^\top$ produces a candidate pair $(h,u)$, but not every such candidate corresponds to a valid $U$ because the recovered $U$ may be rank-deficient. The rank-$1$ extraction is just a necessary condition that we follow with the full-rank check in Phase~4 to filter out spurious candidates.
\end{remark}

\begin{algorithm}[htbp]
\fontsize{8}{9}\selectfont
\LinesNumbered
\DontPrintSemicolon
\caption{Phase 2: construct the linearized kernel}
\label{alg:guessV_phase2}
\KwIn{$M=m+\ell_1$, $N=n+\ell_2$, $(C_i)_{i=1}^{mk}\subset \F_q^{M\times N}$ and
$V\in \F_q^{N\times (k+1)}$.}
\KwOut{$A(V)$ and matrices $K_1,\ldots,K_\rho\in F^{(k+1)\times M}$
whose vectorizations form a basis of $\ker_{\F_q} A(V)$.}
$\mathsf{Mat}\gets0_{mk\times (k+1)M}$\;
\For{$i=1,\ldots,mk$}{
  $D_i\gets C_iV\in \F_q^{M\times (k+1)}$\;
  Set $\mathsf{Mat}[i,(j-1)M+s]\gets D_i[s,j]$ for
  $1\le j\le k+1$, $1\le s\le M$\;
}
$(b_1,\ldots,b_\rho)\gets\text{a basis of }\ker_{\F_q} \mathsf{Mat}$\;
$K_t\gets\operatorname{mat}_{k+1,M}(b_t)$ for $t=1,\ldots,\rho$\;
\Return $(\mathsf{Mat},(K_t)_{t=1}^{\rho})$\;
\end{algorithm}
\FloatBarrier

\subsubsection{Phase 3: Rank-1 solution extraction.}
A naive approach is to compute the Gr\"obner basis of the homogeneous quadratic system. However, the computation is often expensive because it tries to produce a complete symbolic elimination description of the ideal, and for quadratic determinantal systems the intermediate degree and number of S-polynomial reductions can grow quickly. As the system at the end of Phase 2 has symmetries, which will be explained shortly, we use the eigenvalue method for solving zero-dimensional polynomial systems (see \cite{cox2005using}, Chapter~2, Section~4) which significantly accelerates the solving process\footnote{This method is sometimes called Stickelberger-Eigenvalue method. An interesting history of the method can be found in the note of David A. Cox in \cite{cox2021stickelberger}.}.

The first observation is that the quadratic system generated from $2 \times 2$ minors is projective because if $\alpha$ is a solution,
so is $\lambda\alpha$ for every $\lambda\in \F_{q^m}^{\times}$.  For an
index $t_0$ such that $\alpha[t_0]\ne0$, we set $\alpha[t_0]=1$.
After this, we get an affine system, with $\rho-1$ unknowns.  

Let $\mathcal M$ be the degree-$2$ Macaulay matrix obtained from linearization of the quadratic system. Experimentally, we observed that \textit{at degree $2$ the $\mathcal M$ nullity dimension is already $m$ and remains $m$ at
higher degrees} (see Appendix \ref{appendix:macaulay}). Thus, we propose the following conjecture whose supporting arguments will be provided in \Cref{sec:prob_rank_def_U}.
\begin{conjecture} 
\label{conj:dim_ker_m}  
Let
$
I\subseteq \F_q[\alpha[1],\dots,\alpha[\rho-1]]
$
be the ideal generated by the $2\times2$ minors, and let $A=\F_q[\alpha[1],\dots,\alpha[\rho-1]]/I$ then
  \[\dim_{\F_q} A=m.\]
\end{conjecture}

The second observation is that, firstly, the $2\times2$ minors have coefficients in $\F_q$ because
the kernel basis matrices $K_t$ have entries in $\F_q$; and secondly we only need
one rank-$1$ point whose coordinates lie in $\F_{q^m}$.  
Therefore it is
more natural to keep the expensive quotient computation over $\F_q$ and let the less expensive extraction of the desired
$\F_{q^m}$-solution for the final eigenvector evaluation step.  

The third observation is that as $A(V)$ has entries in $\F_q$, it commutes with $\sigma$. Denote  $\sigma(W)$ as the entrywise application of $\sigma$ to $W$, then we have
\[
A(V)\sigma(W)=\sigma(A(V)W).
\]
In consequence, $\ker_{\F_{q^m}}(A(V)) $ is Frobenius-stable. In other words, if
\[
P=(p_1,\dots,p_{\rho-1})\in \mathbf V_{\F_{q^m}}(I),
\]
then, for $0\le i<m$,
\[
P^{[q^i]} := (p_1^{q^i},\dots,p_{\rho-1}^{q^i})
   \in \mathbf V_{\F_{q^m}}(I).
\]
Thus $P,P^{[q]},\dots,P^{[q^{m-1}]}$ are Frobenius-conjugate zeros of
$I$.  Under \Cref{conj:dim_ker_m}, the quotient
has degree $m$, exactly the size of the orbit.  This structure enables us to compute Macaulay row reductions over the small field
$\F_q$ avoiding the cost of linear algebra over $\F_{q^m}$.\\

\noindent\textbf{The eigenvalue method for zero-dimensional
ideals.} First, we recall a well known result in commutative algebra:
\begin{theorem}[Finiteness Theorem]
  \label{theo:finiteness}
  
\noindent Given a field $F$ and polynomials $f_1,\ldots,f_s \in F[x_1,\ldots,x_n]$, the system
\begin{equation}
\label{syspolys}
f_1 = \ldots = f_s = 0
\end{equation}
has finitely many solutions over the algebraic closure $\overline{F}$ of $F$ if and only if 
\[
A = F[x_1,\dots,x_n]/\langle f_1,\dots,f_s\rangle
\]
has finite dimension over $F$.
\end{theorem}
The proof can be found in the proof of Theorem 6, Chapter 5, Section 3 of \cite{cox1997ideals}.
\begin{theorem}[Eigenvalue Theorem]
\label{theo:ETbasic}
Let $f \in F[x_1,\dots,x_n]$ be a polynomial with a
multiplication map
$M_f : A \longrightarrow A.
$
When $\dim_F A < \infty$, the eigenvalues of $M_f$ are the values of $f$ at the finitely many solutions of \Cref{syspolys} over $\overline{F}$.
\end{theorem}
The proof can be found in the proof of Theorem 4.5, Chapter 2, Section 4 of \cite{cox2005using}.

This finiteness is
essential for the eigenvalue method. \Cref{theo:ETbasic}
applies to multiplication maps in the quotient algebra
\[
A=\F_q[\alpha[1],\dots,\alpha[\rho-1]]/I,
\]
and a finite multiplication matrix $M_f$ exists only when this quotient has a
finite $\F_q$-basis.  By the Finiteness Theorem, this is equivalent to the
normalized ideal $I$ being zero-dimensional, or equivalently to the affine
rank-$1$ equations having only finitely many solutions over
$\overline\F_{q}$.  If $I$ had a positive-dimensional component, then $A$
would be infinite-dimensional. In our application, this finiteness hypothesis is supplied by
\Cref{conj:dim_ker_m}.  

Next, we build the multiplication matrix of $A$ with a separating coordinate (for exact definition and properties, refer to Proposition 4.1, Section 4, Chapter 2, in \cite{cox2005using}), i.e., compute the linear map: 
\begin{align*}
M_{\alpha[s]}:A &\longrightarrow A\\
f&\longmapsto \alpha[s] f.
\end{align*}
Let $
\mathcal B=(b_1,\dots,b_{m})$ be a basis of $A$. For an eigenvalue $\lambda\in\F_{q^m}$, an eigenvector of
$M_{\alpha[s]}^{\mathsf T}$, normalized so that its entry corresponding to
$1$ is $1$, gives
$
(b_1(P),b_2(P),\dots,b_{m}(P))
$.
After normalization, this reconstructs
$\alpha(P)$, hence the rank-$1$ matrix $W(\alpha)$. The cost of this eigenvalue extraction is polynomial. 

 Algorithms~\ref{alg:guessV_phase3} and~\ref{alg:residual_eigen} summarize
the two stages of Phase 3. Except for Step 18 of Algorithm~\ref{alg:residual_eigen} being done over $\F_{q^m}$, all operations take place over $\F_q$. This explains the advantage of the eigenvalue method in comparison with a naive Gr\"obner-basis computation over $\F_{q^m}$. The total cost of Phase 3 is dominated by the cost of computing the basis of $A$ over $\F_q$ and is estimated as
\begin{equation}\label{eq:minrank_cost}
\mathcal O\!\left(\tbinom{\rho+1}{2}^{\omega}\right)
=\mathcal O\!\left(\tbinom{m+(k+1)\ell_1+1}{2}^{\omega}\right)
\end{equation}
where $\omega$ is the linear algebra exponent.

\begin{remark}
\label{remark:fglm}
In Algorithm~\ref{alg:guessV_phase3}, we can reduce the number of variables from $\rho-1$ to $m-1$ before applying the eigenvalues extraction. Recall that the Gaussian elimination on the degree-$2$ Macaulay matrix $\mathcal M$
produces a linear system of rank
$(\rho-1)-(m-1)$ on the degree-$1$ variables $\alpha$. Under \Cref{conj:dim_ker_m} the kernel of $\mathcal M$ has
dimension exactly $m$, so the linearized system has a
$d = (m-1)$-dimensional affine solution family
\begin{equation*}\label{eq:affine_parametrization_main}
\alpha \;=\; \tilde\alpha_0 \;+\; \sum_{i=1}^{d} c[i]\,\tilde\alpha_i \; \text{ for } c=(c[1],\dots,c[d])\in \Fqm^d,
\end{equation*}
where the vectors $\tilde\alpha_i\in \F_q^{\rho-1}$ are obtained from the row-reduced
echelon form of $\mathcal M$ and hence have entries in $\F_q$. The
residual quadratic system $$J\subset \F_q[c[1],\dots,c[d]]$$ is generated by quadratic
polynomials in $c[1],\dots,c[d]$ whose coefficients lie in $\F_q$. For the correct $V$ the zero-set $\mathbf V_{\F_q}(J')$ is finite of
cardinality $\dim_{\F_q} A = m$, where $A=\F_q[c[1],\dots,c[d]]/J'$ is
the affine coordinate ring. This count matches the kernel nullity
observed experimentally in Appendix~\ref{appendix:macaulay}.
\end{remark}

\par\noindent
\begin{minipage}{\textwidth}
\SetAlgoSkip{}
\begin{algorithm}[H]
\fontsize{8}{9}\selectfont
\LinesNumbered
\DontPrintSemicolon
\caption{$\mathsf{RankOne}$}
\label{alg:guessV_phase3}
\KwIn{Kernel basis matrices $K_1,\ldots,K_\rho$.}
\KwOut{$(\alpha,W)$ with $W=\sum_t\alpha[t]K_t$,
$\rank_{\F_{q^m}} W=1$, or $\perp$.}
\For{$t_0=1,\ldots,\rho$}{
  \tcp{Normalize.}
  $\alpha[t_0]\gets1$\nllabel{step:rankone_normalize}\;
  $\mathcal M\gets\mathsf{Mac}_2(\sum_t\alpha[t]K_t)$\nllabel{step:rankone_linearize}\tcp{Build the degree-2 Macaulay matrix.}
  $\mathcal T\gets\mathsf{EigenvaluesSolve}(\mathcal M)$\nllabel{step:rankone_extract}\;
  \tcp{Verify and return an accepted candidate.}
  \ForEach{$x\in\mathcal T$}{
    Reconstruct $\alpha$ by inserting $\alpha[t_0]=1$ into
    $x$\nllabel{step:rankone_reconstruct}\;
    \lIf{$\rank_{\F_{q^m}} W=1$}{\Return $(\alpha,W)$\nllabel{step:rankone_accept}}
  }
}
\Return $\perp$\;
\end{algorithm}

\vspace{6pt}
\begin{algorithm}[H]
\fontsize{8}{9}\selectfont
\LinesNumbered
\DontPrintSemicolon
\caption{$\mathsf{EigenvaluesSolve}$}
\label{alg:residual_eigen}
\KwIn{A degree-$2$ Macaulay matrix $\mathcal M$.}
\KwOut{A list of vectors, or $\varnothing$.}
\tcp{Compute the quotient basis.}
Split $\mathcal M=[A\mid B]$, where columns of $A$ correspond to the constant and degree-$1$ monomials, and columns of $B$ correspond to the degree-$2$ monomials\nllabel{step:eigen_reduce}\;
Row-reduce $B$ to get $\widetilde B$; applying the same row operations to $A$ to get $\widetilde A$ \;
$\mathcal M \gets [\widetilde{A}\mid \widetilde{B}]$\;
Split $\widetilde B$ into $\widetilde B_0$, containing its zero rows,
and $\widetilde B_1$, containing its nonzero rows\;
Let $L_1,L_2$ be the rows of $\widetilde A$ corresponding to
$\widetilde B_0,\widetilde B_1$, respectively\;
$L_1 \gets [c \mid D]$ where $c$ is the first column of $L_1$, containing the constant terms; $D$ its remaining $\rho-1$ columns, containing the coefficients of $\alpha[t]_{t\ne t_0}$, in increasing index order\;
Compute $x$ in $Dx=-c$ \;
Compute $y$ in $\widetilde B_1 y=-L_2\begin{pmatrix}1\\x\end{pmatrix}$\nllabel{step:eigen_relations}\;
$\mathcal B\gets\{1\}$\nllabel{step:eigen_basis}\;
\ForEach{monomial $b$ in $x$ or $y$}{
  \lIf{$b$ is linearly independent of $\mathcal B$}{append $b$ to $\mathcal B$}
}
$\delta\gets|\mathcal B|$ \tcp{ From Conjecture \ref{conj:dim_ker_m}, we expect  $\delta = m$.}
\lIf{no such basis is determined at degree $2$, or $\delta=0$}{\Return $\varnothing$\nllabel{step:eigen_basis_check}}
\tcp{Compute the eigenvalues.}
Write each entry of $x$ as a linear combination of
$\mathcal B=(1,b_2,\ldots,b_\delta)$, and store its coefficients in
the corresponding row of $C$, so that
$x=C(1,b_2,\ldots,b_\delta)^{\mathsf T}$\nllabel{step:eigen_coordinates}\;
Choose an entry $x[s]$ of $x$\nllabel{step:eigen_choose}\;
\ForEach{monomial $b_j$ in $\mathcal B$, with $b_1=1$}{
  Write $x[s]b_j$ as a linear combination of $\mathcal B$, and
  store its coefficients in column $j$ of
  $T\in\F_q^{\delta\times\delta}$\nllabel{step:eigen_multiplication}\;
}
Compute the eigenvalues $\lambda\in \F_{q^m}$ of $T^{\mathsf T}$ and corresponding eigenvectors $v$\;
\tcp{Recover coordinate vectors $x$.}
$\mathcal T\gets[\,]$\;
\ForEach{computed eigenvector $v$ with $v[1]\ne0$}{
  $v\gets v/v[1]$\nllabel{step:eigen_normalize}\tcp{Make the constant entry $1$.}
  $x^\star\gets Cv$\nllabel{step:eigen_readout}\tcp{Recover the original coordinates.}
  Append $x^\star$ to $\mathcal T$\nllabel{step:eigen_append}\;
}
\Return $\mathcal T$\;
\end{algorithm}
\end{minipage}
\par\medskip
\FloatBarrier

\subsubsection{Phase 4: Factorization and reconstruction.}
Once a rank-1 matrix $W$ is found, we resconstruct the row compression matrix $U$ using Algorithm \ref{alg:guessV_phase4}. The distinguisher succeeds if the algorithm recovers a \emph{full-rank} matrix $U$. 
\begin{remark}
  Given a successful rank-1 extraction from Phase 3, the probability of $U$ begin full-rank is overwhelmingly high (Step 5 in Algorithm \ref{alg:guessV_phase4}). This probability is analyzed in Section \ref{sec:prob_rank_def_U}.
\end{remark}

\begin{algorithm}[htbp]
\fontsize{8}{9}\selectfont
\LinesNumbered
\DontPrintSemicolon
\caption{Phase 4: reconstruct and verify $U$}
\label{alg:guessV_phase4}
\KwIn{$(C_i)_{i=1}^{mk}$, $V$, a rank-$1$ matrix
$W\in \F_{q^m}^{n'\times M}$, $\Psi$.}
\KwOut{$U$, or $\perp$.}
Choose $(j_0,s_0)$ with $W[j_0,s_0]\ne0$\;
$u[s]\gets W[j_0,s]$ for $s=1,\ldots,M$\;
$h[j]\gets W[j,s_0]/u[s_0]$ for $j=1,\ldots,n'$
\tcp*{$W=hu^{\mathsf T}$}
$U\gets[\Psi(u[1])\mid\cdots\mid\Psi(u[M])]\in F^{m\times M}$\;
\lIf{$\rank_F U<m$}{\Return $\perp$}
\Return $U$\;
\end{algorithm}

\subsection{Guess-$U$-solve-$V$ direction}
\label{sec:guess_U_solve_V_rank1}
The same idea can be run in the opposite direction.  One may guess the row compression $U$ and solve algebraically for a full column compression
$
V\in\F_q^{N\times m}$.  For simplicity, and as in the concrete EGMC
parameter sets considered below, we specialize this direction to $n=m$, so
$N=m+\ell_2$.  
\subsubsection{Phase 1: Guessing $U$.}
Here, the probability of guessing the correct $U$ is $$\frac{|\Gr(m,m)|}{|\Gr(m,m+\ell_1)|}
\approx q^{-m\ell_1}.$$

\subsubsection{Phase 2: Algebraic modelling.}
Given a valid guess $U \in\F_q^{m\times(m+\ell_1)}$, let $C_1,\dots,C_{mk}$ be an $\F_q$-basis of the public code
$\mathcal C$. For each public basis matrix set
$
E_i=UC_i=(E_i[a,r])\in\F_q^{m\times N}.
$
We want to find a full-rank column compression
\[
V=
\begin{bmatrix}
v_1\\
\vdots\\
v_N
\end{bmatrix}
\in\F_q^{N\times m}
\text{ where } v_1,\dots,v_N\in\F_q^m,
\]
such that the compressed code $UC_iV$ is annihilated by a full parity-check
matrix of the hidden $[m,k]$ Gabidulin code. As in the previous direction, we write the unknown rows of $V$ into extension
field coordinates:
\[
y[r]=\Psi^{-1}(v_r)\in\Fqm
\quad \forall r=1,\dots,N .
\]
Fix a row-wise extension of $\Psi$, let
$\row_{\Fqm}$ denote the inverse row expansion on $m$ rows:
\begin{align*}
\row_{\Fqm}=(\Psi^{-1})^{\oplus m}:
\F_q^{m\times m} &\longrightarrow \Fqm^m\\
\begin{bmatrix}
x_1\\
\vdots\\
x_m
\end{bmatrix}
&\longmapsto
\bigl(\Psi^{-1}(x_1),\dots,\Psi^{-1}(x_m)\bigr).
\end{align*}
\noindent Since $E_i[a,r]\in\F_q$, we have
\[
\bigl(\row_{\Fqm}(E_iV)\bigr)_a
   = \sum \limits_{r=1}^{N} E_i[a,r]\,y[r].
\]
Let $s=m-k$, and let $H\in\Fqm^{s\times m}$ be a parity-check matrix.  Since left multiplication by any matrix in
$\GL_s(\Fqm)$ gives the same parity-check space, we normalize
$
H=[I_s\mid B]$ where $B=(B[b,a])\in\Fqm^{s\times k}$.
The parity-check condition is
\begin{align}
  \label{eq:matrixH_bilinear}
&H\,\row_{\Fqm}(E_iV)^{\mathsf T}=0,
\quad  \forall i=1,\dots,mk \notag \\
\implies &\sum_{r=1}^{N} E_i[b,r]\,y[r]
+
\sum_{a=1}^{k}\sum_{r=1}^{N}
E_i[s+a,r]\,B[b,a]\,y[r]
=0 
\end{align}for $1\le i\le mk$ and $1\le b\le s$. Let us denote
$
x[b,a,r]=B[b,a]\,y[r]
$ for $1\le b\le s, 1\le a\le k, 1\le r\le N $.
Then \eqref{eq:matrixH_bilinear} becomes the linear system
\begin{equation}
\label{eq:matrixH_linear}
\sum_{r=1}^{N} E_i[b,r]\,y[r]
+
\sum_{a=1}^{k}\sum_{r=1}^{N}
E_i[s+a,r]\,x[b,a,r]
=0
\quad \forall 
i=1,\dots,mk, \; b=1,\dots,s.
\end{equation}
All coefficients lie in $\F_q$, so its solution space over $\Fqm$ is the
scalar extension of its kernel over $\F_q$. Hence, we are back to the problem of finding a rank-$1$ element in the kernel of dimension 
\[
\rho_H
=
m+\ell_2\bigl(1+k(m-k)\bigr).
\]

\subsubsection{Phase 3: Rank-$1$ solution extraction.}

We therefore compute a basis $b_1,...,b_{\rho_H}$ of the $\F_q$-kernel of
\eqref{eq:matrixH_linear}, set
$
K_t=\operatorname{mat}_{1+sk,N}(b_t)\in\F_q^{(1+sk)\times N},
$
and set to zero all the $2\times2$ minors of
\[
Z(\alpha)=\sum_{t=1}^{\rho_H}\alpha[t] K_t .
\]
After one projective normalization, the residual quotient is also conjectured to be
zero-dimensional with supporting experimental evidence given in \Cref{appendix:macaulay_guessU}, so the same eigenvalue method as in
the guess-$V$-solve-$U$ direction applies.  The rank-$1$ extraction here uses a Macaulay matrix with
$\binom{\rho_H+1}{2}$ columns after projective normalization, hence cost
\[
\mathcal O\!\left(\binom{\rho_H+1}{2}^{\omega}\right)
\]
over $\mathbb{F}_q$ where $\omega$ is the linear algebra exponent.

\subsubsection{Phase 4: Factorization and reconstruction.}
Once a rank-$1$ point $Z$ is found, the first row is $
y=(y[1],\dots,y[N]).
$
The candidate column compression is recovered by
\[
V=
\begin{bmatrix}
\Psi(y[1])^{\mathsf T}\\
\vdots\\
\Psi(y[N])^{\mathsf T}
\end{bmatrix}
\in\F_q^{N\times m}.
\]
If $y[r_0]\ne0$, the entries of the normalized parity-check matrix are
recovered from
$
B[b,a]=x[b,a,r_0]/y[r_0].
$
The candidate is accepted only if $V$ has rank $m$ over~$\F_q$.

\section{Recovering the equivalent secret keys}
\label{sec:code_extension}
When the distinguisher succeeds, we are given a matrix code $
\mathcal{B}_0=U_0\mathcal C V_0
$
where $\mathcal C$ is the public EGMC code, $U_0\in \F_q^{m\times (m+\ell_1)}$ and
$V_0\in \F_q^{(n+\ell_2)\times n'}$ are full-rank matrices.
In the guess-$U$-solve-$V$ direction, we have $n' =n = m$, thus $
\mathcal{B}_0$ is equivalent to the matrix representation of an $[n,k]$ Gabidulin code over $\F_{q^m}$.  Nevetheless, in the guess-$V$-solve-$U$ direction, $n' = k+1 < n$, thus $
\mathcal{B}_0
$ is equivalent to the matrix representation of just a $[k+1,k]$ Gabidulin code
over $\F_{q^m}$. This is a punctured code of the $[n,k]$ Gabidulin code
over $\F_{q^m}$. Our goal now is to extend this short punctured Gabidulin code into a full equivalent $[n,k]$ Gabidulin description of the row-compressed public code. This section therefore applies to only the guess-$V$-solve-$U$ direction.

It should be noted that a punctured generalized Reed-Solomons code can be extended by 1 digit while keeping the MDS property \cite{seroussi2003mds}. In an EGMC instance, the structure of the punctured Gabidulin code, together with the full row-compressed public code, instead lets us recover the hidden clean column space and thereby return to length~$n$. Indeed, the punctured code reveals the hidden copy of $\F_{q^m}$ acting on the rows and hence the Frobenius map in the recovered coordinates. As a result, we can recover the equivalent secret keys after two steps. First, in Section~\ref{sec:reconstruct_gabidulin}, we reconstruct the punctured $[k+1,k]$ Gabidulin code in standard generator-matrix form. Finally, as the generator-matrix of Gabidulin codes is special, Section~\ref{sec:extension} shows an algorithm to extend the $[k+1,k]$ Gabidulin code  back to the full $[n,k]$ Gabidulin code. The whole process is entirely polynomial-time linear algebra over $\F_q$ and $\F_{q^m}$.

\subsection{Reconstruct the $[k+1,k]$ Gabidulin code over $\F_{q^m}$ in its standard form}
\label{sec:reconstruct_gabidulin}
Assume that we are given
$
\mathcal{B}_0=U_0 \mathcal C V_0 \subseteq \Mat_{m\times (k+1)}(\F_q)
$
of $\F_q$-dimension $mk$, and that $\mathcal{B}_0$ is equivalent to the
matrix representation of a $[k+1,k]$ Gabidulin code over $\F_{q^m}$. Using the \texttt{BasisRecover} Algorithm from \cite{PorwalWL2026ACDGV}, we can find in polynomial time a basis $\gamma$ of $\F_{q^m}/\F_q$ and an evaluation vector
$
\mathbf g'=({g'}_1,\dots,{g'}_{k+1})\in \F_{q^m}^{k+1}
$
with $\F_q$-linearly independent coordinates such that
$
\mathcal B_0=\Psi_\gamma\!\bigl(\mathcal G_k(\mathbf g')\bigr)$. 
We modified the \texttt{BasisRecover} algorithm so that it also returns the standard generator matrix in Moore form (see Definition 7, \cite{augot2021rank}) of the $[k+1,k]$ Gabidulin code, which will be used in the next step to extent the code to its full length.

\[
M_{k+1,k}(\boldsymbol g')=
\begin{bmatrix}
{g'}_1 & {g'}_2 & \cdots & {g'}_{k+1}\\
{g'}_1^q & {g'}_2^q & \cdots & {g'}_{k+1}^q\\
\vdots & \vdots & \ddots& \vdots\\
{g'}_1^{q^{k-1}} & {g'}_2^{q^{k-1}} & \cdots & {g'}_{k+1}^{q^{k-1}}
\end{bmatrix}.
\]

\begin{algorithm}[htbp]
\fontsize{8}{9}\selectfont
\LinesNumbered
\DontPrintSemicolon
\caption{\texttt{BasisRecover} \cite{PorwalWL2026ACDGV}}
\label{alg:basis_recover_moore}
\KwIn{An $\F_q$-basis $B_1,\ldots,B_{mk}$ of
$\mathcal B_0\subseteq\F_q^{m\times (k+1)}$.}
\KwOut{$\gamma$, $\Phi_\gamma$, $\mathbf g'$, and
$M_k(\mathbf g')$, or $\perp$.}
\BlankLine
$\mathcal L \gets\operatorname{Stab}(\mathcal B_0)
=\{T\in\F_q^{m\times m}:T\mathcal B_0\subseteq\mathcal B_0\}$\;
\lIf{$\dim_{\F_q}\mathcal L\ne m$}{\Return $\perp$}
sample $T\in\mathcal S$ until its minimal polynomial
$f_T(X)\in\F_q[X]$ has degree $m$, using at most $\left(1-2q^{-m \over 2}+2q^{-m}\right)^{-1} \le 2$ trials\;
\lIf{no such $T$ is found, or $f_T$ is reducible}{\Return $\perp$}
choose $0\ne w\in\F_q^m$ and set
$R_T\gets[\,w\mid Tw\mid\cdots\mid T^{m-1}w\,]\in\GL_m(\F_q)$\;
find a root $\vartheta\in\Fqm$ of $f_T$\;
$\gamma\gets
\left(1, \vartheta, \vartheta^2, \dots, \vartheta^{m-1}\right) R_T^{-1} \in \F_{q^m}^m$\;
$\mathcal C_\gamma\gets\Psi_\gamma^{-1}(\mathcal B_0)
\subseteq\Fqm^{k+1}$\;
apply the standard Overbeck algorithm to
$\mathcal C_\gamma$ and obtain a vector
$\mathbf g'\in\Fqm^{k+1}$ with $\F_q$-linearly independent coordinates\;
$M_{k+1,k}(\mathbf g') \gets
[\,\mathbf g';(\mathbf g')^{[q]};\ldots;
(\mathbf g')^{[q^{k-1}]}\,]$\;
\Return $(\gamma,\mathbf g',M_{k+1,k}(\mathbf g'))$\;
\end{algorithm}

\subsection{Extending the punctured code to a full length-$n$ Gabidulin code}
\label{sec:extension}
Let $\mathcal D=U_0\mathcal C\subseteq\F_q^{m\times (n+\ell_2)} $ be the row-compressed public code. From previous step, we have: 
\[
 \Psi_\gamma^{-1}(\mathcal D)V_0=\mathcal G_k(\mathbf g')
 \subseteq\F_{q^m}^{k+1}.
\]
\noindent Next, we define the puncture map
\begin{align*}
 \label{eq:puncture_isomorphism}
 \pi:\Psi_\gamma^{-1}(\mathcal D)
 &\longrightarrow\mathcal G_k(\mathbf g') \\
  b&\longmapsto bV_0 \;.
\end{align*}
We have $\dim_{\F_q}\Psi_\gamma^{-1}(\mathcal D)$ is at most $mk$ over $\F_q$ because it is a linear image of the public code $\mathcal C$. The punctured Gabidulin code $\mathcal G_k(\mathbf g')$ also has dimension $mk$ over $\F_q$. Since $\pi$ is surjective, it is an $\F_q$-linear isomorphism. Thus, we can compute $b_0,...,b_{k-1} \in \Psi_\gamma^{-1}(\mathcal D) \subset \F_{q^m}^{n+\ell_2}$ uniquely such that 
\begin{equation}
 \label{eq:lifted_moore_rows}
 b_iV_0=(\mathbf g')^{[q^i]} ,
\end{equation}

\noindent then we seek a full-rank matrix $V \in \F_q^{(n+\ell_2) \times n}$ and a vector $\mathbf g \in \F_{q^m}^n$ of $\F_q$-rank $n$ satisfying such that:
\begin{equation}
 \label{eq:V_moore_rows}
 b_iV=(\mathbf g)^{[q^i]},
\end{equation}

Recall that the construction of EGMC and the assumption that $U_0$ has been correctly guessed ensure the existence of a solution $V$ having the form 
\begin{equation*}
  V = Q^{-1} \begin{bmatrix} * \\ 0 \end{bmatrix} \text{ where $Q\in \GL_{n+\ell_2}(\F_q)$ and}
\end{equation*}
\begin{equation*}
 \begin{bmatrix}b_0\\\vdots\\b_{k-1}\end{bmatrix} Q^{-1}=\begin{bmatrix}M_{n,k}(\mathbf g)& |& \Delta \end{bmatrix} \text{  where } \Delta \text{ is a random matrix in $\F_{q^m}^{k\times\ell_2}$.}
\end{equation*}

\noindent Note that for any vectors $x\in\F_q^{n+\ell_2}$, we have $$\langle b_0,x \rangle^{q^i}=\langle b_0^{[q^i]},x \rangle.$$ Thus,  for $\delta_i$ being rows of $\Delta$, writing $b_i$ in $\F_q$ basis via $\Psi_\gamma$ gives us:
\begin{equation*}
 \begin{bmatrix} \Psi_\gamma(b_1 - b_0^{[q]})\\ \vdots \\ \Psi_\gamma(b_{k-1} - b_0^{[q^{k-1}]}) \end{bmatrix}Q^{-1}=\left[
\begin{array}{c|c}
0 &
\begin{matrix}
\Psi_\gamma(\delta_1-\delta_0^{[q]})\\
\vdots\\
\Psi_\gamma(\delta_{k-1}-\delta_0^{[q^{k-1}]})
\end{matrix} 
\end{array}
\right] = \left[0 | \widetilde{\Delta}\right].
\end{equation*}
\noindent For $m(k-1) > \ell_2$, with high probablity, $\widetilde{\Delta}$ is full-rank $\ell_2$ over $\F_q$, thus 
\[
\rank \begin{bmatrix} \Psi_\gamma(b_1 - b_0^{[q]})\\ \vdots \\ \Psi_\gamma(b_{k-1} - b_0^{[q^{k-1}]}) \end{bmatrix} = \ell_2.
\]
\noindent Combining with Equation~\eqref{eq:V_moore_rows}, the column space of $V$ is
\[
\mathcal L
=
\left\{
v\in \F_q^{n+\ell_2}:
\langle b_{i}-b_{0}^{[q^i]}, v\rangle=0 \text{ for }
i=1,\dots,k-1
\right\},
\] 
and its dimension is $n$. Compute any $\F_q$-basis $v_1,\ldots,v_n$ of $\mathcal L$, and setting
\begin{align*}
 \label{eq:extended_V_and_g}
 V&=[\,v_1\mid\cdots\mid v_n\,]\in\F_q^{N\times n}\\
 \mathbf g&=\bigl(\langle b_0,v_1\rangle,\ldots,\langle b_0,v_n\rangle\bigr)
 \in\F_{q^m}^n.
\end{align*} gives us the final solutions.

Algorithm~\ref{alg:extend_support} summarizes the extension procedure. The cost of the algorithm is polynomial-time linear algebra over $\F_q$ and
$\F_{q^m}$.

\begin{algorithm}[htbp]
\fontsize{8}{9}\selectfont
\LinesNumbered
\DontPrintSemicolon
\caption{\texttt{ExtendSupport}: extension to a full Gabidulin code}
\label{alg:extend_support}
\KwIn{A code $\mathcal C\subseteq\F_q^{(m+\ell_1)\times(n+\ell_2)}$; full-rank
matrices $U_0\in\F_q^{m\times(m+\ell_1)}$ and
$V_0\in\F_q^{(n+\ell_2)\times(k+1)}$; and
$\gamma,\mathbf g'$ satisfying
$U_0\mathcal C V_0=\Psi_\gamma(\mathcal G_k(\mathbf g'))$.}
\KwOut{A full-rank $V\in\F_q^{(n+\ell_2)\times n}$ and an evaluation
vector $\mathbf g\in\F_{q^m}^n$ of $\F_q$-rank $n$, or $\perp$.}
\BlankLine
$N\gets n+\ell_2$ and
$d_s\gets\Psi_\gamma^{-1}(U_0C_s)\in\F_{q^m}^N$ for
$1\leq s\leq mk$\;
$\widetilde{\mathcal D}\gets
\Span_{\F_q}\{d_1,\ldots,d_{mk}\}$\;
$\mathcal A\gets
\Span_{\F_q}\{d_1V_0,\ldots,d_{mk}V_0\}$\;
\lIf{$\dim_{\F_q}\mathcal A\ne mk$}{\Return $\perp$}
\For{$i=0,\ldots,k-1$}{
  solve for $b_i \in\widetilde{\mathcal D} \subset \F_{q^{m}}^N$ the system $b_iV_0 = (\mathbf g')^{[q^i]}$;
}
$\mathcal L \gets \ker_{\F_q}
[\,\Psi_\gamma(b_1-b_0^{[q]})^{\mathsf T}\mid\cdots\mid
\Psi_\gamma(b_{k-1}-b_0^{[q^{k-1}]})^{\mathsf T}\,]^{\mathsf T}$\;
\lIf{$\dim_{\F_q}\mathcal L\ne n$}{\Return $\perp$}
choose an $\F_q$-basis $v_1,\ldots,v_n$ of $\mathcal L$ and set
$V\gets[\,v_1\mid\cdots\mid v_n\,]$\;
\For{$j=1,\ldots,n$}{
  $g_j\gets \langle b_0,v_j \rangle$\;
}
$\mathbf g\gets(g_1,\ldots,g_n)$\;
\lIf{$\rank_{\F_q}(\mathbf g)\ne n$}{\Return $\perp$}
\lIf{$\widetilde{\mathcal D}V\ne\mathcal G_k(\mathbf g)$
as $\F_q$-linear spaces}{\Return $\perp$}
\Return $(V,\mathbf g)$\;
\end{algorithm}

\section{Applications to EGMC-McElice
and EGMC-Niederreiter}
\label{sec:decoding_ciphertext}
Here, using the recovered equivalent trapdoor, we explain polynomial-time decryption
algorithms for the EGMC-McEliece/Niederreiter schemes.
\subsection{Breaking the EGMC-McEliece scheme}
Let
$
\text{pk}=(M_1,\dots,M_{km})
$
be a basis of the public matrix code $\mathcal C'_{\mathrm{mat}}$, and let a ciphertext be
$
Y=\sum_{i=1}^{km}\mu_i M_i + E$ where $\rank(E)\le r.
$
Applying the recovered $U,V$ yields
\[
UYV
=
\sum_{i=1}^{km}\mu_i (UM_iV) + UEV.
\]
Note that $\rank(UEV) \le \rank(E) \le r$ and since $M_i\in \mathcal C'_{\mathrm{mat}}$, each $UM_iV$ belongs to
$
U\,\mathcal C'_{\mathrm{mat}}\,V
=
\Psi_\gamma\!\bigl(\mathcal G_k(\boldsymbol g)\bigr),
$
hence
\[
UYV \in \Psi_\gamma\!\bigl(\mathcal G_k(\boldsymbol g)\bigr) + \{Z:\rank(Z)\le r\}.
\]
Therefore, provided $r$ is within the decoding radius of the recovered
Gabidulin code, one can decode $UYV$ and recover the unique codeword
$
C^\star=\sum_{i=1}^{km}\mu_i (UM_iV)
\in \Psi_\gamma\!\bigl(\mathcal G_k(\boldsymbol g)\bigr).
$
It remains to recover the message $\mu$. Since
\[
C^\star=\sum_{i=1}^{km}\mu_i (UM_iV),
\]
and the matrices $UM_iV$ are known, $\mu$ is obtained by solving the linear
system above over $\F_q$. The solution is unique because
$
\dim_{\F_q}\bigl(U\,\mathcal C'_{\mathrm{mat}}\,V\bigr)=mk.
$

\subsection{Breaking the EGMC-Niederreiter scheme}

Let $\bar H$ be a parity-check matrix of the public matrix code $\mathcal C'_{\mathrm{mat}}$. In
the Niederreiter variant, the ciphertext is a syndrome
$
c=\bar H\,\bar e^{\,t},
$
where $\bar e=\text{Unfold}(E)$ for some matrix
$
E\in \F_q^{(m+\ell_1)\times(n+\ell_2)}
$ where $
\rank(E)\le r.
$
From the public syndrome $c$, first compute any vector
$
\bar y\in \F_q^{(m+\ell_1)(n+\ell_2)}
$
such that
$
\bar H\,\bar y^{\top}=c,
$
and let $Y=\text{Fold}(\bar y)$. Then $Y$ belongs to the affine space 
$
Y \in E + \mathcal C'_{\mathrm{mat}},
$
so there exists some $C\in \mathcal C'_{\mathrm{mat}}$ such that
$
Y=C+E.
$
Applying the recovered $U,V$ yields
\[
UYV = UCV + UEV,
\]
where $UCV\in \Psi_\gamma(\mathcal G_k(\boldsymbol g))$ and
$\rank(UEV)\le r$. Therefore, Gabidulin decoding recovers the projected error
$
E^\star = UEV.
$

To reconstruct the full error $E$, observe that $E$ must satisfy both:
\[
UEV = E^\star
\quad\text{and}\quad
\bar H\,\text{Unfold}(E)^\top = c.
\]
This is a linear system in the entries of $E$ over $\F_q$. Solving it
recovers $E$, and hence the plaintext
$
\mu=\text{Unfold}(E).
$
\section{Arguments for \Cref{conj:dim_ker_m}}
\label{sec:prob_rank_def_U}

This section gives supporting arguments for \Cref{conj:dim_ker_m} in the
guess-$V$-solve-$U$ direction; the guess-$U$-solve-$V$ direction can be argued
analogously.  First, in \Cref{prop:dimA_lower} we will prove that \[\dim_{\F_q}A\ge m.\] This also explains why the points in this hidden orbit pass the
full-rank check in Phase~4.  Second, in \Cref{prop:dimA_upper} we will prove that
$$\Pr[\dim_{\F_q}A>m]\le q^{(1-k)(m-1)+\ell_1}.$$
\subsection{Full-rank preservation from Frobenius orbits}
\label{subsec:rank_preservation_orbit}
We recall necessary background to prove that \[\dim_{\F_q}A\ge m.\] 
Firstly, recall that 
$
I\subseteq \F_q[\alpha[1],\dots,\alpha[\rho-1]]
$
is the ideal generated by the $2\times2$ minors of
$
W(\alpha)=\sum_t\alpha[t]K_t .
$
For $P=(p_1,\dots,p_{\rho-1})\in \mathbf V_{\F_{q^m}}(I)$, $W(P)$ is a rank-$1$ element of $A(V)$.  We denote $[W(P)]$ as its class modulo multiplication by
$\F_{q^m}^{\times}$.  The projective rank-$1$ solution set is denoted as
\[
\mathcal S(V)
= \Bigl\{[W(P)] : P\in \mathbf V_{\F_{q^m}}(I)\Bigr\}.
\]

Recall that the whole $\ker_{\F_{q^m}}(A(V)) $ is Frobenius-stable so does $\mathcal S(V)$.  We therefore decompose it into disjoint
Frobenius orbits:
\begin{equation}
\mathcal S(V)=\bigcup_{\nu\in T} \mathcal O_\nu ,
\label{eq:orbit_decomp}
\end{equation}
where
\[
\mathcal O_\nu
=
\{[W(P_\nu)],[W(P_\nu^{[q]})],\dots,
  [W(P_\nu^{[q^{d_\nu-1}]})]\}.
\]
Here $P_\nu$ is a representative of the orbit and $d_\nu$ is the smallest
positive integer such that
$[W(P_\nu^{[q^{d_\nu}]})]=[W(P_\nu)]$.  We introduce two following lemmas to explain how
the orbit size controls the rank of the reconstructed matrix $U$.

\begin{lemma}
\label{lem:orbit_subfield}
Let $W=h u^{\mathsf T}$ be a rank-$1$ matrix over $\F_{q^m}$.  Assume that
$[\sigma^d(W)]=[W]$ for some $d>0$ and $g=\gcd(d,m)$.  If
$u[s_0]\ne0$, then, for every $s$,
\[
\frac{u[s]}{u[s_0]}\in \F_{q^g}.
\]
In particular, if $d$ is the Frobenius orbit size of
$[W]$, then $d\mid m$.
\end{lemma}

\begin{proof}
The equality $[\sigma^d(W)]=[W]$ means that
$\sigma^d(W)=\lambda W$ for some $\lambda\in\F_{q^m}^{\times}$.  Since
$W$ has rank $1$, its row space is the line generated by $u^{\mathsf T}$;
therefore $\sigma^d(u)$ is proportional to $u$.  Hence, for every $s$,
\[
\sigma^d\!\left(\frac{u[s]}{u[s_0]}\right)
 =
\frac{u[s]}{u[s_0]}.
\]
The fixed field of $\sigma^d$ inside $\F_{q^m}$ is
$\F_{q^g}$, with $g=\gcd(d,m)$, proving the first claim.  Then because the Frobenius group
$\mathrm{Gal}(\F_{q^m}/\F_q)=\langle\sigma\rangle$ is cyclic of order
$m$, we have $d \mid m$. 
\end{proof}

\begin{lemma}
\label{lem:rankU_le_d}
Let $u=(u[1],\dots,u[b])\in\F_{q^m}^b$ and $U(u)$ its matrix form over $\F_q$.  If there exist
$d\mid m$ such that
$u[s]\in \F_{q^d}$ for every $s$, then $\rank_{\F_q}(U(u))\le d$.
\end{lemma}
\begin{proof}
The subfield $\F_{q^d}$ is a $d$-dimensional $\F_q$-vector space, hence every
column of $U(u)$ lies in the $d$-dimensional subspace
$\Psi(\F_{q^d})\subseteq\F_q^m$.
Therefore \[\rank_{\F_q}(U(u))\le d.\]
\end{proof}

The next proposition shows that, for a valid guess of $V$, the set
$\mathcal S(V)$ always contains the hidden Frobenius orbit and that this
orbit has the expected size $m$.

\begin{proposition}
\label{prop:dimA_lower}
\label{prop:frobenius_orbit_full_rank}
For any valid guess of $V$, \[\dim_{\F_q} A \ge m.\]  Moreover, the
rank-$1$ points in the hidden Frobenius orbit all reconstruct a full-rank
row compression in Phase~4.
\end{proposition}
\begin{proof}
Let $W^*=h^*(u^*)^{\mathsf T}$ be the rank-$1$ kernel matrix obtained from
the hidden row compression $U^*$.  Since
$\ker_{\F_{q^m}}(A(V))$ is Frobenius-stable, for every $i=0,\dots,m-1$,
\[
[\sigma^i(W^*)]\in \mathcal S(V).
\]
We first prove that these $m$ projective points are distinct.  Suppose, to the
contrary, that $[\sigma^d(W^*)]=[W^*]$ for some $0<d<m$.  Choose
$s_0$ with $u^*[s_0]\ne0$, and set $g=\gcd(d,m)$.  By
\Cref{lem:orbit_subfield}, we have $u^* \in \F_{q^g}$.  Since $g<m$, \Cref{lem:rankU_le_d} gives
$\rank_{\F_q}U^*\le g<m$.  This contradicts the fact that 
\[
\rank_{\F_q} U^*
 =
\dim_{\F_q}\Span_{\F_q}\{u^*[1],\dots,u^*[m+\ell_1]\}
 =m .
\]
Therefore the hidden Frobenius orbit has size exactly $m$. After normalization, the same orbit gives $m$ distinct zeros of
the ideal $I$ over $\F_{q^m}$ hence $\dim_{\F_q}A\ge m$.
\noindent
Finally, any point in this orbit has the form
\[
\lambda\,\sigma^i(W^*)
  =(\lambda\,\sigma^i(h^*))\,\sigma^i(u^*)^{\mathsf T}
\]
for some $\lambda\in\F_{q^m}^{\times}$.  Frobenius and multiplication by a
nonzero scalar are $\F_q$-linear automorphisms of $\F_{q^m}$, so they
preserve the $\F_q$-span dimension of the entries of $u^*$.  Thus every
rank-$1$ point in the hidden orbit reconstructs a matrix $U(u)$ of rank
$m$ in Phase~4. 
\end{proof}
\subsection{The probability of unwanted solutions}
\label{subsec:unwanted_solutions}
Here, we prove the following proposition:

\begin{proposition}
\label{prop:dimA_upper}
Assume that $m$ is prime, under the random-subspace model described below,
\[
\Pr\!\bigl[\dim_{\F_q} A > m\bigr] \;\le\; q^{(1-k)(m-1)+\ell_1}.
\]
For the concrete parameter sets the exponent is dominated by $-(k-1)(m-1)$, so this bound is negligible.
\end{proposition}
\begin{proof}
From Lemmas~\ref{lem:orbit_subfield} and~\ref{lem:rankU_le_d}, any solution
whose $u$-coordinates are projectively defined over a proper subfield can only
give a rank-deficient compression. Thus a full-rank recovered $U(u)$ can only come from a solution whose coordinates generate the full
$\F_q$-space $\F_{q^m}$.  We consider two cases:
\begin{enumerate}
\item If we assume that $m$ is prime, which is true for concrete parameters. Obviously, the only divisors of $m$ are $1$ and $m$. Therefore every orbit in \Cref{eq:orbit_decomp} has size either $1$ or $m$. Furthermore, size-$m$ orbits correspond to solutions not defined over any proper subfield (``generic''),
while size-$1$ orbits correspond to solutions defined over $\F_q$ (projectively Frobenius-fixed).
By Lemma~\ref{lem:rankU_le_d}, size-$1$ orbits satisfy $\rank_{\F_q}(U(u))\le 1$, hence (except the trivial all-zero $u$ which cannot happen for rank-1 $W$) typically $\rank_{\F_q}(U(u))=1$. Thus, for prime $m$, orbit-size-$1$ solutions are the main systematic source of rank-deficient outputs; the hidden size-$m$ orbit is full rank by \Cref{prop:frobenius_orbit_full_rank}. 

\item If we assume that $m$ is composite, although it is not the case with the concrete parameters. Then, orbit sizes can be any divisor $d\mid m$. In that case, orbit size $d$ typically yields $\rank_{\F_q}(U)\le d$ by Lemma~\ref{lem:rankU_le_d}. In particular, size-$1$ orbits correspond to $\F_q$-defined solutions and often yield $\rank_{\F_q}(U)=1$, size-$d$ orbits (proper divisors $1<d<m$) correspond to $\F_{q^d}$-defined solutions and often yield $\rank_{\F_q}(U)=d$, and size-$m$ orbits correspond to ``full field'' solutions and allow $\rank_{\F_q}(U)=m$.
\end{enumerate}
\noindent
Now we consider the following question: \textit{For random EGMC public code, random $V$ from Phase 1, if one selects a rank-1 solution $[W]$ \emph{uniformly at random} from $\mathcal{S}(V)$, then what is the
\[
\Pr(\rank_{\F_q}(U)<m\mid V) 
=
\frac{\#\{[W]\in \mathcal S(V): \rank_{\F_q}(U(W))<m\}}{\#\mathcal S(V)} \; ?
\]}
\noindent
Firstly, from Equation \eqref{eq:orbit_decomp} and the fact that $m$ is prime for concrete parameters, the total number of projective rank-1 solutions admits the decomposition
\begin{equation*}
\#\mathcal S(V) = m\cdot a_m(V) + a_1(V),
\label{eq:count_prime}
\end{equation*}
where $a_m(V)$ counts size-$m$ orbits and $a_1(V)$ counts size-$1$ orbits.
We denote $a=k+1$ and $b=m+\ell_1$ and 
\[
L(V)=
\{\,W\in \Mat_{a\times b}(\F_{q^m}) :
\operatorname{vec}_{a,b}(W)\in \ker_{\F_{q^m}} A(V)\,\}.
\]
The $\F_q$-rational part of the kernel is
\[
L_0(V)= L(V)\cap \Mat_{a\times b}(\F_q).
\]
\noindent
Let $r_0=\dim_{\F_q} L_0(V)$ and the projective rank-1 locus over $\F_q$ is
\[
\Sigma_0 = \{[W]\in \mathbb P(\Mat_{a\times b}(\F_q)):\ \rank_{\F_q}(W)=1\}.
\]
Then, in the $m$ prime case, size-$1$ Frobenius orbits correspond precisely to $\F_q$-rational rank-1 points:
\[
a_1(V)=\#\bigl(\mathbb P(L_0(V))\cap \Sigma_0\bigr).
\]
\noindent
Because every nonzero rank-1 matrix $W\in \Mat_{a\times b}(\F_q)$ can be written as $W=xy^\top$ for some nonzero column vector $x$ and row vector $y^\top$ and such a factorization is unique up to multiplication by a scalar\footnote{This is known as the Segre embedding.},
the number of projective rank-1 matrices is
\begin{align*}
&\#\Sigma_0 = \#\mathbb P^{a-1}(\F_q) \times \#\mathbb P^{b-1}(\F_q) = \left(\frac{q^a-1}{q-1}\right) \left(\frac{q^b-1}{q-1}\right)
\end{align*} while number of projective matrices of size $a \times b$ and projective $(r_0-1)$-subspaces are approximately
\begin{align*}
&\#\mathbb P(\Mat_{a\times b}(\F_q)) = \frac{q^{ab}-1}{q-1}\notag\\
&\#\mathbb P(L_0(V)) = \frac{q^{r_0}-1}{q-1}.\notag
\end{align*}
\noindent
We model $\mathbb P(L_0(V))$ as a uniformly random projective $(r_0-1)$-subspace of $\mathbb P^{ab-1}(\F_q)$ independent of $\Sigma_0$. Under this model, the expected number of rank-1 $\F_q$-points is
\begin{align}
\label{eq:Ea1_asymp_general}
\mathbb E[a_1(V)]
&\approx
\#\mathbb P(L_0(V))\cdot
\frac{\#\Sigma_0}{\#\mathbb P(\Mat_{a\times b}(\F_q))}\nonumber\\
&=
\frac{q^{r_0}-1}{q-1}\cdot
\frac{(q^a-1)(q^b-1)}{(q-1)(q^{ab}-1)} \nonumber\\
&\approx q^{\,r_0+a+b-ab-2}.
\end{align}
Over $\F_{q^m}$, if $A(V)$ has full row rank $mk$ (over $\F_{q^m}$, equivalently over $\F_q$), then
\[
r = \dim_{\F_{q^m}} L(V)=ab-mk =(k+1)(m+\ell_1)-mk = m+(k+1)\ell_1.
\]
Recall that in Phase 3 of the Hybrid Distinguisher, we have proven that the kernel over the extension field $\F_{q^m}$ is simply the scalar extension of the kernel over $\F_q$. Therefore
\[
r_0 = \dim_{\F_q}(\ker_{\F_q} A) = \dim_{\F_{q^m}}(\ker_{\F_{q^m}} A) = r.
\]
For $a=k+1$, $b=m+\ell_1$, $r_0= m+(k+1)\ell_1$,  Equation \eqref{eq:Ea1_asymp_general} gives the expectation:
\begin{equation}
\mathbb E[a_1(V)] \approx q^{(1-k)(m-1)+\ell_1}.
\label{eq:Ea1_special}
\end{equation}
For $k\ge 2$, the exponent is dominated by $-(k-1)(m-1)$, predicting that $a_1(V)$ is very small. Therefore,  Markov's inequality gives
\begin{equation*}
\Pr\big(a_1(V)\ge 1\big)\ \le\ \mathbb E[a_1(V)].
\label{eq:Markov_bound}
\end{equation*}
If $\E[a_1] \ll 1$, we can approximately assume $\Pr (a_1 \ge 1) \approx \E[a_1]$. Hence, under our random-subspace model and when \eqref{eq:Ea1_special} is $\ll 1$, the probability that we obtain base-field rank-1 solutions should be negligible. Empirically, we observe that, with overwhelming probability, $a_m(V) = 1$ for ``valid'' $V$, and the solution set size is
$
\#\mathcal S(V)= m
$ (see Appendix \ref{appendix:macaulay}). Combined with the EGMC assumption $a_m(V)=1$, this gives
\[
\Pr\!\bigl[\#\mathcal{S}(V) > m\bigr]
\;=\;\Pr\!\bigl[a_1(V)\ge 1\bigr]
\;\le\; q^{(1-k)(m-1)+\ell_1},
\]
which gives \[\Pr[\dim_{\F_q} A > m] \le q^{(1-k)(m-1)+\ell_1}.\] 
\end{proof}

\section{Complexity for concrete parameter sets}
\label{sec:complexity}
There are two directions of the distinguisher, thus the overall cost of the attack will be analyzed accordingly.  In the guess-$V$-solve-$U$
direction, the cost of the attack is
\begin{equation}
  \label{eq:guessV_solveU_cost_main}
C_{V\to U} 
=
\mathcal O\!\left(
\underbrace{q^{(k+1)\ell_2}}_{\text{$V$-guess}}
\underbrace{
\tbinom{m+(k+1)\ell_1+1}{2}^{\omega}}_{\text{rank-$1$ extraction}} \right) .
\end{equation}
In the guess-$U$-solve-$V$
direction, the cost of the attack is 
\begin{equation}
  \label{eq:guessU_solveV_cost_main}
C_{U\to V}
=
\mathcal O\!\left(
\underbrace{q^{m\ell_1}}_{\text{$U$-guess}}
\underbrace{\tbinom{m+\ell_2(1+k(m-k))+1}{2}^{\omega}}_{\text{rank-$1$ extraction}}
\right).
\end{equation} 
In both directions, the $\omega = 2.8$ is the linear algebra constant. Our attack is especially strong when $\ell_1=0$ or $\ell_2 = 0$ because
the exponential guessing factor disappears and the attack becomes
polynomial-time. The final
attack estimate is the better of the two directions:
\begin{equation}
\label{eq:hybrid_attack_min_cost}
\min\{\log_2 C_{V\to U},\log_2 C_{U\to V}\},
\end{equation}
and the cost for concrete parameter sets are given in \Cref{tab:attack_cost_summary}.
\paragraph{Acknowledgements}
The author would like to thank Brice Minaud for helpful discussions. This work was supported by the France 2030 ANR Project ANR-22-PECY-003 SecureCompute and the ANR grant PQMC (ANR-25-CE39-7662).
\appendix
\par\noindent
\begin{minipage}{\textwidth}
\setlength{\intextsep}{4pt}
\section{Macaulay diagnostics for quadratic determinantal systems in guess-$V$-solve-$U$ direction}
\label{appendix:macaulay}
\begin{table}[H]
    \captionsetup{font=small,skip=3pt}
    \caption{Small guess-$V$-solve-$U$ quadratic determinantal systems: each degree entry gives the Macaulay rank/column count in corresponding degree $D$. The $\%$ column reports the percentage of full-rank $U$ recoveries over 1000 trials per parameter set.}
    \label{macaulay}
    \centering
    \fontsize{8}{9}\selectfont
    \renewcommand{\arraystretch}{1}
    \setlength{\tabcolsep}{1pt}
    \setlength{\aboverulesep}{0.5pt}
    \setlength{\belowrulesep}{1pt}
    \begin{minipage}[t]{0.485\linewidth}
    \begin{tabular*}{\linewidth}[t]{@{\extracolsep{\fill}}ccccccccc@{}}
        \toprule
        \multicolumn{5}{c}{Parameters} & \multicolumn{3}{c}{Rank / columns} & \multirow{2}{*}{$\%$} \\
        \cmidrule(lr){1-5} \cmidrule(lr){6-8}
        $q$ & $k$ & $m$ & $\ell_1$ & $\ell_2$ & $D=2$ & $D=3$ & $D=4$ &  \\
        \midrule
        2 & 2 & 4 & 1 & 0 & $24/28$ & $80/84$ & $206/210$ & 66 \\
        \midrule
        2 & 2 & 4 & 1 & 1 & $24/28$ & $80/84$ & $206/210$ & 99.5 \\
        \midrule
        2 & 3 & 4 & 1 & 0 & $32/36$ & $116/120$ & $326/330$ & 94.6 \\
        \midrule
        2 & 3 & 4 & 1 & 1 & $32/36$ & $116/120$ & $326/330$ & 100 \\
        \midrule
        2 & 3 & 6 & 1 & 0 & $49/55$ & $214/220$ & $709/715$ & 99 \\
        \midrule
        2 & 3 & 6 & 1 & 1 & $49/55$ & $214/220$ & $709/715$ & 100 \\
        \midrule
        2 & 6 & 7 & 1 & 0 & $98/105$ & $553/560$ & $2373/2380$ & 100 \\
        \midrule
        2 & 6 & 7 & 1 & 1 & $98/105$ & $553/560$ & $2373/2380$ & 100 \\
        \midrule
        2 & 5 & 7 & 1 & 0 & $84/91$ & $448/455$ & $1813/1820$ & 100 \\
        \midrule
        2 & 5 & 7 & 1 & 1 & $84/91$ & $448/455$ & $1813/1820$ & 100 \\
        \midrule
        2 & 4 & 7 & 1 & 0 & $71/78$ & $357/364$ & $1358/1365$ & 100 \\
        \midrule
        2 & 4 & 7 & 1 & 1 & $71/78$ & $357/364$ & $1358/1365$ & 100 \\
        \bottomrule
    \end{tabular*}
    \end{minipage}
    \hfill
    \begin{minipage}[t]{0.485\linewidth}
    \begin{tabular*}{\linewidth}[t]{@{\extracolsep{\fill}}ccccccccc@{}}
        \toprule
        \multicolumn{5}{c}{Parameters} & \multicolumn{3}{c}{Rank / columns} & \multirow{2}{*}{$\%$} \\
        \cmidrule(lr){1-5} \cmidrule(lr){6-8}
        $q$ & $k$ & $m$ & $\ell_1$ & $\ell_2$ & $D=2$ & $D=3$ & $D=4$ &  \\
        \midrule
        2 & 4 & 7 & 2 & 0 & $146/153$ & $962/969$ & $4838/4845$ & 100 \\
        \midrule
        2 & 4 & 7 & 2 & 2 & $146/153$ & $962/969$ & $4838/4845$ & 100 \\
        \midrule
        2 & 4 & 5 & 1 & 0 & $50/55$ & $215/220$ & $710/715$ & 100 \\
        \midrule
        2 & 4 & 5 & 1 & 1 & $50/55$ & $215/220$ & $710/715$ & 100 \\
        \midrule
        2 & 3 & 5 & 1 & 0 & $40/45$ & $160/165$ & $490/495$ & 100 \\
        \midrule
        2 & 3 & 5 & 1 & 1 & $40/45$ & $160/165$ & $490/495$ & 100 \\
        \midrule
        2 & 3 & 5 & 1 & 2 & $40/45$ & $160/165$ & $490/495$ & 100 \\
        \midrule
        2 & 3 & 5 & 1 & 3 & $40/45$ & $160/165$ & $490/495$ & 100 \\
        \midrule
        2 & 2 & 5 & 1 & 0 & $31/36$ & $115/120$ & $325/330$ & 100 \\
        \midrule
        2 & 2 & 5 & 1 & 1 & $31/36$ & $115/120$ & $325/330$ & 100 \\
        \midrule
        2 & 4 & 5 & 2 & 0 & $115/120$ & $675/680$ & $3055/3060$ & 100 \\
        \midrule
        2 & 4 & 5 & 2 & 2 & $115/120$ & $675/680$ & $3055/3060$ & 100 \\
        \bottomrule
    \end{tabular*}
    \end{minipage}
\end{table}
\section{Macaulay diagnostics for quadratic determinantal systems in guess-$U$-solve-$V$ direction}
\label{appendix:macaulay_guessU}
\begin{table}[H]
    \captionsetup{font=small,skip=3pt}
    \caption{Small guess-$U$-solve-$V$ quadratic determinantal systems: each degree entry gives the Macaulay rank/column count in corresponding degree $D$. The $\%$ column reports the percentage of full-rank $V$ recoveries over 1000 trials per parameter set.}
    \label{macaulay_guessU}
    \centering
    \fontsize{8}{9}\selectfont
    \renewcommand{\arraystretch}{1}
    \setlength{\tabcolsep}{1pt}
    \setlength{\aboverulesep}{0.5pt}
    \setlength{\belowrulesep}{1pt}
    \begin{minipage}[t]{0.485\linewidth}
    \begin{tabular*}{\linewidth}[t]{@{\extracolsep{\fill}}cccccccc@{}}
        \toprule
        \multicolumn{4}{c}{Parameters} & \multicolumn{3}{c}{Rank / columns} & \multirow{2}{*}{$\%$} \\
        \cmidrule(lr){1-4} \cmidrule(lr){5-7}
        $q$ & $k$ & $m$ & $\ell_2$ & $D=2$ & $D=3$ & $D=4$ &  \\
        \midrule
        2 & 2 & 4 & 1 & $41/45$ & $161/165$ & $491/495$ & 81.0 \\
        \midrule
        2 & 3 & 4 & 1 & $32/36$ & $116/120$ & $326/330$ & 94.0 \\
        \midrule
        2 & 3 & 4 & 2 & $74/78$ & $360/364$ & $1361/1365$ & 80.0 \\
        \midrule
        2 & 3 & 6 & 1 & $130/136$ & $810/816$ & $3870/3876$ & 100 \\
        \midrule
        2 & 6 & 7 & 1 & $98/105$ & $553/560$ & $2373/2380$ & 100 \\
        \bottomrule
    \end{tabular*}
    \end{minipage}
    \hfill
    \begin{minipage}[t]{0.485\linewidth}
    \begin{tabular*}{\linewidth}[t]{@{\extracolsep{\fill}}cccccccc@{}}
        \toprule
        \multicolumn{4}{c}{Parameters} & \multicolumn{3}{c}{Rank / columns} & \multirow{2}{*}{$\%$} \\
        \cmidrule(lr){1-4} \cmidrule(lr){5-7}
        $q$ & $k$ & $m$ & $\ell_2$ & $D=2$ & $D=3$ & $D=4$ &  \\
        \midrule
        2 & 4 & 5 & 1 & $50/55$ & $215/220$ & $710/715$ & 100 \\
        \midrule
        2 & 3 & 5 & 1 & $73/78$ & $359/364$ & $1360/1365$ & 100 \\
        \midrule
        2 & 2 & 5 & 1 & $73/78$ & $359/364$ & $1360/1365$ & 100 \\
        \midrule
        2 & 4 & 5 & 2 & $115/120$ & $675/680$ & $3055/3060$ & 100 \\
        \bottomrule
    \end{tabular*}
    \end{minipage}
\end{table}
\end{minipage}
\par\clearpage
\section{Encryption schemes based on EGMC}
\label{appendix:McEliece_Niederreiter_encryption_framework}
\begin{figure}[H]
    \centering
    \scriptsize
    \captionsetup{font=scriptsize,skip=4pt}
    \setlength{\fboxsep}{3pt}
    \setlength{\abovedisplayskip}{3pt}
    \setlength{\belowdisplayskip}{3pt}
    \setlength{\abovedisplayshortskip}{2pt}
    \setlength{\belowdisplayshortskip}{2pt}
    \renewenvironment{itemize}{%
        \begin{list}{}{%
            \setlength{\leftmargin}{7pt}%
            \setlength{\labelwidth}{4pt}%
            \setlength{\labelsep}{3pt}%
            \setlength{\topsep}{2pt}%
            \setlength{\itemsep}{0.8pt}%
            \setlength{\parsep}{0pt}%
            \setlength{\partopsep}{0pt}%
        }%
    }{\end{list}}
    \begin{minipage}[t]{0.485\textwidth}
        \vspace{0pt}
        \fbox{%
        \begin{minipage}[t]{\dimexpr\linewidth-2\fboxsep-2\fboxrule\relax}
            \raggedright
            \setlength{\parskip}{2pt}
        
        \textbf{KeyGen} ($1^\lambda$):
        \begin{itemize}
            \item[-] Select a random $[m, k]_{q^m}$ Gabidulin code, with an efficient algorithm capable of decoding up to $\left\lfloor \frac{m-k}{2} \right\rfloor$ errors.
            \item[-] Sample uniformly at random a basis $\gamma \xleftarrow{\$} \mathcal{B}(\mathbb{F}_{q^m})$.
            \item[-] Compute a basis $(\boldsymbol{A}_1, \dots, \boldsymbol{A}_{km})$ of the code $\Psi_\gamma(\mathcal{G})$.
            \item[-] For $i$ in range 1 to $km$, sample uniformly at random: $\boldsymbol{R}_i \xleftarrow{\$} \mathbb{F}_q^{m \times \ell_2}$, $\boldsymbol{R}'_i \xleftarrow{\$} \mathbb{F}_q^{\ell_1 \times m}$ and $\boldsymbol{R}''_i \xleftarrow{\$} \mathbb{F}_q^{\ell_1 \times \ell_2}$.
            \item[-] Define the matrix code $\mathcal{C}_{mat}$ as explained in Definition 16.
            \item[-] Sample uniformly at random matrices $\boldsymbol{P} \xleftarrow{\$} \mathbf{GL}_{m+\ell_1}(\mathbb{F}_q)$ and $\boldsymbol{Q} \xleftarrow{\$} \mathbf{GL}_{m+\ell_2}(\mathbb{F}_q)$.
            \item[-] Define the code $\mathcal{C}'_{mat} = \boldsymbol{P}\mathcal{C}_{mat}\boldsymbol{Q}$.
            \item[-] Let $\mathcal{B} = (\boldsymbol{M}_1, \dots, \boldsymbol{M}_{km})$ be a basis of $\mathcal{C}'_{mat}$.
            \item[-] Return $\text{pk} = \mathcal{B}$ and $\text{sk} = (\mathcal{G}, \gamma, \boldsymbol{P}, \boldsymbol{Q})$.
        \end{itemize}

        \textbf{Encrypt}(pk, $\mu$):
        \begin{itemize}
             \item[] \textit{Input:} $\text{pk} = (\boldsymbol{M}_1, \dots, \boldsymbol{M}_{km})$, $\mu \in \mathbb{F}_q^{km}$.
             \item[-] Sample uniformly at random a matrix $\boldsymbol{E} \in \mathbb{F}_q^{(m+\ell_1)\times(m+\ell_2)}$ such that rank $\boldsymbol{E} \le r$.
             \item[-] Return $\boldsymbol{Y} = \sum_{i=1}^{km} \mu_i \boldsymbol{M}_i + \boldsymbol{E}$.
        \end{itemize}

        \textbf{Decrypt}(sk, $\boldsymbol{Y}$):
        \begin{itemize}
            \item[-] Compute $\boldsymbol{P}^{-1}\boldsymbol{Y}\boldsymbol{Q}^{-1}$. Truncate the $\ell_1$ last rows to obtain $\boldsymbol{M} \in \mathbb{F}_q^{m \times (m+\ell_2)}$.
            \item[-] Compute $\Psi_\gamma^{-1}(\boldsymbol{M}) \in \mathbb{F}_{q^m}^{m+\ell_2}$. Let $\boldsymbol{y}$ be the first $m$ entries.
            \item[-] Apply the decoding algorithm of $\mathcal{G}$ on the word $\boldsymbol{y}$, it returns an error vector $\boldsymbol{e} \in \mathbb{F}_{q^m}^m$.
            \item[-] Compute $\boldsymbol{E}' = \Psi_\gamma(\boldsymbol{e})$.
            \item[-] Solve the linear system:
            \[
            \boldsymbol{Y} - (\boldsymbol{E}' | \boldsymbol{N})\boldsymbol{Q} = \sum_{i=1}^{km} \mu_i \boldsymbol{M}_i
            \]
            whose $(k+\ell_2)m$ unknowns in $\mathbb{F}_q$ are the $\mu_i$ and the remaining part of the error $\boldsymbol{N}$ of size $m \times \ell_2$.
            \item[-] Return $\hat{\mu}$ the solution of the above system.
        \end{itemize}
        \end{minipage}%
        }
        \caption{EGMC-McEliece encryption scheme}
        \label{fig:egmc_mceliece}
    \end{minipage}
    \hfill
    \begin{minipage}[t]{0.485\textwidth}
        \vspace{0pt}
        \fbox{%
        \begin{minipage}[t]{\dimexpr\linewidth-2\fboxsep-2\fboxrule\relax}
            \raggedright
            \setlength{\parskip}{2pt}
        
        \textbf{KeyGen} ($1^\lambda$):
        \begin{itemize}
            \item[-] Select a random $[m, k]_{q^m}$ Gabidulin code, with an efficient algorithm capable of decoding up to $\left\lfloor \frac{m-k}{2} \right\rfloor$ errors.
            \item[-] Sample uniformly at random a basis $\gamma \xleftarrow{\$} \mathcal{B}(\mathbb{F}_{q^m})$.
            \item[-] Compute a basis $(\boldsymbol{A}_1, \dots, \boldsymbol{A}_{km})$ of the code $\Psi_\gamma(\mathcal{G})$.
            \item[-] For $i$ in range 1 to $km$, sample uniformly at random: $\boldsymbol{R}_i \xleftarrow{\$} \mathbb{F}_q^{m \times \ell_2}$, $\boldsymbol{R}'_i \xleftarrow{\$} \mathbb{F}_q^{\ell_1 \times m}$ and $\boldsymbol{R}''_i \xleftarrow{\$} \mathbb{F}_q^{\ell_1 \times \ell_2}$.
            \item[-] Define the matrix code $\mathcal{C}_{mat}$ as explained above.
            \item[-] Sample uniformly at random matrices $\boldsymbol{P} \xleftarrow{\$} \mathbf{GL}_{m+\ell_1}(\mathbb{F}_q)$ and $\boldsymbol{Q} \xleftarrow{\$} \mathbf{GL}_{m+\ell_2}(\mathbb{F}_q)$.
            \item[-] Define the code $\mathcal{C}'_{mat} = \boldsymbol{P}\mathcal{C}_{mat}\boldsymbol{Q}$.
            \item[-] Compute $\bar{\boldsymbol{H}} \in \mathbb{F}_q^{((m+\ell_1)(m+\ell_2)-mk)\times(m+\ell_1)(m+\ell_2)}$ a parity check matrix of $\mathcal{C}'_{mat}$.
            \item[-] Return $\text{pk} = \bar{\boldsymbol{H}}$; $\text{sk} = (\gamma, \mathcal{G},\boldsymbol{P}, \boldsymbol{Q})$
        \end{itemize}

        \textbf{Encrypt}(pk, $\mu$):
        \begin{itemize}
             \item[] \textit{Input:} $\text{pk} = \bar{\boldsymbol{H}}$, a message $\mu \in \mathbb{F}_q^{(m+\ell_1)(m+\ell_2)}$ such that rank $\text{Fold}(\mu) \le r$.
             \item[-] For every integer $i$ from 1 to $(m+\ell_1)(m+\ell_2)$, let $\boldsymbol{h}_i$ be the $i$-th column of $\bar{\boldsymbol{H}}$.
             \item[-] Return $\boldsymbol{c} = \sum_{i=1}^{(m+\ell_1)(m+\ell_2)} \mu_i \boldsymbol{h}_i^\top$.
        \end{itemize}

        \textbf{Decrypt}(sk, $\boldsymbol{c}$):
        \begin{itemize}
            \item[] \textit{Input:} $\text{sk} = (\gamma, \mathcal{G}, \boldsymbol{P}, \boldsymbol{Q})$, $\boldsymbol{c} \in \mathbb{F}_q^{(m+\ell_1)(m+\ell_2)-mk}$.
            \item[-] Find any $\bar{\boldsymbol{y}} \in \mathbb{F}_q^{(m+\ell_1)(m+\ell_2)}$ such that $\boldsymbol{c} = \sum_{i=1}^{(m+\ell_1)(m+\ell_2)} \bar{y}_i \boldsymbol{h}_i^\top$.
            \item[-] Let $\boldsymbol{Y} = \text{Fold}(\bar{\boldsymbol{y}})$. Compute $\boldsymbol{P}^{-1}\boldsymbol{Y}\boldsymbol{Q}^{-1}$. Truncate the $\ell_1$ last rows to obtain $\boldsymbol{M} \in \mathcal{M}_{m \times (m+\ell)}(\mathbb{F}_q)$.
            \item[-] Let $\boldsymbol{y}$ be the first $m$ coordinates of $\Psi_\gamma^{-1}(\boldsymbol{M}) \in \mathbb{F}_{q^m}^{m+\ell_2}$.
            \item[-] Let $\boldsymbol{y} = \boldsymbol{m}\boldsymbol{G} + \boldsymbol{e}$, with rank$(\boldsymbol{e}) \le r$. Apply the decoding algorithm of $\mathcal{G}$ on the word composed of the $m$ first coordinates.
            \item[-] Let $\boldsymbol{E} = \Psi_\gamma(\boldsymbol{e})$. Then $\boldsymbol{E} = (\boldsymbol{E}' | \boldsymbol{N}) \in \mathcal{M}_{m \times (m+\ell_2)}(\mathbb{F}_q)$ is a matrix whose the $nm$ coefficients of $\boldsymbol{E}'$ are known, and $\boldsymbol{N}$ is the remaining part of the error.
            \item[-] Solve the linear system $\boldsymbol{c} = \sum_{i=1}^{(m+\ell_1)(m+\ell_2)} \bar{e}_i \boldsymbol{h}_i^\top$ to find the $\ell m$ values of $\boldsymbol{N}$, with $\bar{\boldsymbol{e}} = \text{Unfold}(\boldsymbol{E})$ with unknown some coefficients.
            \item[-] Return $\hat{\mu} = \text{Unfold}(\boldsymbol{E})$.
        \end{itemize}
        \end{minipage}%
        }
        \caption{EGMC-Niederreiter encryption scheme}
        \label{fig:egmc_niederreiter}
    \end{minipage}
\end{figure}
\clearpage
\bibliographystyle{splncs04}
\bibliography{mybibliography}
\end{document}